%% file: ms.tex
\documentclass[10pt, conference, compsocconf, final]{IEEEtran}
\input{tex/0-preamble}

\begin{document}
\title{Network Clocks: Detecting the Temporal Scale of Information Diffusion}

\author{
\IEEEauthorblockN{\begin{tabular*}{0.81\textwidth}{@{\extracolsep{\fill} }c c c}
Daniel J. DiTursi\IEEEauthorrefmark{1}\IEEEauthorrefmark{2} & Gregorios A. Katsios\IEEEauthorrefmark{1} & 
Petko Bogdanov\IEEEauthorrefmark{1}\end{tabular*}}

\IEEEauthorblockA{\begin{tabular*}{0.72\textwidth}{@{\extracolsep{\fill} }c c}
\parbox[t]{0.35\textwidth}{\centering\IEEEauthorrefmark{1}Department of Computer Science\\
State University of New York at Albany\\
Albany, NY  12222 } &
\parbox[t]{0.35\textwidth}{\centering\IEEEauthorrefmark{2}Department of Computer Science\\
Siena College\\
Loudonville, NY  12211} 
\end{tabular*}}
}
\maketitle

\input{tex/00-abstract}

\input{tex/1-intro}
\input{tex/2-prelim}

\input{tex/3-problem}

\input{tex/4-sol}

\input{tex/5-exp}

\input{tex/11-related}
\input{tex/6-conclude}

{\footnotesize
\bibliographystyle{abbrvnat}
\bibliography{ref/career,ref/dynclust,ref/references} 
}

\end{document}

%% file: tex/0-preamble.tex
\pdfoutput=1
\usepackage[utf8]{inputenc}
\usepackage[numbers]{natbib}
\usepackage{graphicx}
\usepackage{algorithm}
\usepackage{algpseudocode}
\usepackage{mathtools}
\usepackage{amsfonts}
\usepackage{amsmath}
\usepackage{amsthm}
\usepackage{subfigure}
\usepackage{dirtytalk}
\usepackage{multirow}
\usepackage{xspace}
\usepackage{pbox}
\usepackage{soul}
\usepackage[bookmarks=false]{hyperref}
\usepackage{url}
\usepackage[usenames, dvipsnames]{color}
\usepackage{tikz}
\usepackage{tikzscale}
\usepackage{subfigure}
\usetikzlibrary{calc}
\usetikzlibrary{shapes}
\definecolor{pinegreen}{cmyk}{0.92,0,0.59,0.25}
\definecolor{royalblue}{cmyk}{1,0.50,0,0}
\definecolor{lavander}{cmyk}{0,0.48,0,0}
\definecolor{violet}{cmyk}{0.79,0.88,0,0}
\tikzstyle{cblue}=[circle, draw, thin,fill=cyan!20, scale=1]
\tikzstyle{cred}=[circle, draw, thin,fill=red!20, scale=1]
\tikzstyle{cgreen}=[circle, draw, thin,fill=green!20, scale=1]
\tikzstyle{qgre}=[rectangle, draw, thin,fill=green!20, scale=0.8]
\tikzstyle{rpath}=[thin, red, dashed, opacity=0.1]
\tikzstyle{legend_isps}=[rectangle, rounded corners, thin, 
                       fill=gray!20, text=blue, draw]
                        
\tikzstyle{labelt}=[rectangle, rounded corners, thin,
                           top color= white,bottom color=green!25,
                           black]
\tikzstyle{labelb}=[rectangle, rounded corners, thin,
                          top color= white,bottom color=cyan!25,
                          minimum width=2.5cm, minimum height=0.8cm,
                          royalblue]
\tikzstyle{labelv}=[rectangle, rounded corners, thin,
                          top color= white,bottom color=lavander!25,
                          minimum width=2.5cm, minimum height=0.8cm,
                        violet]

\newtheorem{theorem}{Theorem}

\newtheorem{definition}{Definition}
\newtheorem{corollary}{Corollary}

\DeclareMathOperator*{\argmax}{argmax}

\newcommand{\eat}[1]{}

\newcommand{\intv}[2]{\langle #1,#2\rangle}

\newcommand{\vsa}{\vspace*{-0.1cm}}
\newcommand{\vsb}{\vspace*{-0.2cm}}
\newcommand{\vsc}{\vspace*{-0.4cm}}

\newcommand{\threepartdef}[6]
{
	\left\{
		\begin{array}{lll}
			#1 & \mbox{if } #2 \\
			#3 & \mbox{if } #4 \\
			#5 & \mbox{if } #6
		\end{array}
	\right.
}

\newcommand{\moc}{\text{OC}}
\newcommand{\oc}{$\moc$\xspace}

\newcommand{\mkoc}{\text{kOC}}
\newcommand{\koc}{$\mkoc$\xspace}

\newcommand{\mdp}{\text{DP}}
\newcommand{\dmp}{$\mdp$\xspace}

\newcommand{\mgr}{\text{GR}}
\newcommand{\gr}{$\mgr$\xspace}

\newcommand{\grdp}{$\mgr\mdp$\xspace}

\newcommand{\grsq}{$\mgr^2$\xspace}

\algnewcommand{\IIf}[1]{\State\algorithmicif\ #1\ \algorithmicthen}
\algnewcommand{\EndIIf}{\unskip\ \algorithmicend\ \algorithmicif}

%% file: tex/00-abstract.tex
\begin{abstract}
Information diffusion models typically assume a discrete timeline in which an information token spreads in the network. Since users in real-world networks vary significantly in their intensity  and periods of activity, our objective in this work is to answer: \emph{How to determine a temporal scale that best agrees with the observed information propagation within a network?}

A key limitation of existing approaches is that they aggregate the timeline into fixed-size windows, which may not fit all network nodes' activity periods. We propose the notion of a {\em heterogeneous\em} network clock: a mapping of events to discrete timestamps that best explains their occurrence according to a given cascade propagation model. We focus on the widely-adopted independent cascade (IC) model and formalize the optimal clock as the one that maximizes the likelihood of all observed cascades. The single optimal clock (\oc) problem can be solved exactly in polynomial time. However, we prove that learning multiple optimal clocks (\koc), corresponding to temporal patterns of groups of network nodes, is NP-hard. We propose scalable solutions that run in almost linear time in the total number of cascade activations and discuss approximation guarantees for each variant. Our algorithms and their detected clocks enable improved cascade size classification (up to $8\%$ F1 lift) and improved missing cascade data inference ($0.15$ better recall). We also demonstrate that the network clocks exhibit consistency within the type of content diffusing in the network and are robust with respect to the propagation probability parameters of the IC model.
\end{abstract}

%% file: tex/1-intro.tex
\section{Introduction}

Models for information propagation in social networks assume a discrete timeline over which the information spreads~\protect\cite{goldenberg2001using,granovetter1978threshold}. At the same time, real-world network data is collected at different temporal resolutions e.g. minutes, hours, weeks. Hence, a common step before analyzing real-world data is to define a temporal scale to aggregate events in the network into discrete snapshots. Such aggregations make the implicit assumption that there exists a single temporal resolution that works well for all network nodes and periods of activity. Empirical studies of social network activity, however, demonstrate a wide variance in activity levels across users and across time~\protect\cite{macropol2013act,tagarelli2015time,yang2014finding}, thus challenging the above assumption. \emph{How can we account for heterogeneity in user activity over time when analyzing information propagation data from real-world networks?} 

To address the above problem, we propose the problem of detecting data-driven {\em heterogeneous\em} temporal aggregations for information diffusion under the widely-adopted independent cascade (IC) model~\protect\cite{goldenberg2001using,goldenberg2001talk}. We formalize the problem as finding the temporal aggregation that maximizes the likelihood of observed cascades according to the IC model. We allow this aggregation to vary across time and across users, driven by the likelihood of observed data. We refer to such an optimal temporal aggregation as a \emph{network clock}: a discrete timeline with time steps determined by user activity as opposed to regular intervals of ``wall-clock'' time. 

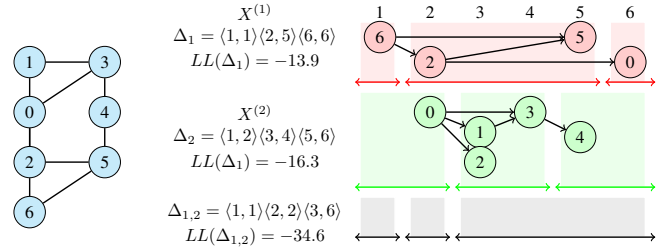
\begin{figure}[t]
\centering
 \resizebox{0.48\textwidth}{!}{%
 \input{fig/example.tikz}
}
\caption{ \footnotesize An example network and two cascades $X^{(1)}$ and $X^{(2)}$ with their original measured activation times between $1$ and $6$. Assuming an independent cascade with spontaneous activation probability $p_e=0.001$ and a neighbor propagation probability $p_n=0.1$, the best clock for $X^{(1)}$ is $\Delta_1$ with a log-likelihood $LL=-13.9$ and the best clock for $X^{(2)}$ is $\Delta_2$ with $LL=-16.3$. An optimal clock that maximizes the likelihood of both cascades is $\Delta_{1,2}$, which is suboptimal for the individual cascades.}
\label{fig:ex}
\vsc\vsb
\end{figure}

The example in Fig.~\ref{fig:ex} involves a small graph of $7$ nodes (left) and two cascades $X^{1}$ and $X^{2}$ (right), which are observed over the duration of $6$ discrete time steps. A cascade is a sequence of activation events with their originally measured times define the \emph{original time resolution} of observed cascades. On the right, cascades are presented as directed graphs where edges denote the possibility of influence from an activated node to a network neighbor which is activated at a later time (the horizontal position of nodes corresponds to their original time of activation). Using the original time resolution for cascade $X^{1}$ will render nodes $5$ and $0$ being as spontaneous activations, since IC requires that potential influence nodes are activated in the immediately preceding time step. Such view of time will result in a low IC likelihood as activations through the network have typically higher-likelihood. Thus, using the original time resolution as a network clock will poorly capture how nodes influenced each other. 

Alternatively, the clock $\Delta_1$ aggregates the middle 4 time ticks into a single time step, i.e. the activations of nodes $2$ and $5$ are treated as simultaneous, with the activation of $6$ preceding them and that of $0$ immediately following. As all activations except the first are now explained by the immediately prior influence of a neighbor, this clock maximizes the likelihood of the observed activations being produced by an IC process.

In our running example (Fig.~\ref{fig:ex}), the clocks $\Delta_1$ and $\Delta_2$ that optimize the two cascades  separately are both different from each other and from the clock that maximizes the likelihood of both cascades considered together ($\Delta_{1,2}$). Imposing the same overall clock on all users and all cascade instances may be overly restrictive. Different users might engage in the information propagation at different times depending on their interest, overall activity, time zone, etc. In addition, different information, e.g. videos vs. text, may spread at different rates and at different times. Hence, we also consider the problem of finding multiple clocks for a single dataset.

Knowledge of the network clocks can improve a variety of applications. Based on past propagation, one can predict likely times of future activity and allocate computational resources to handle the user traffic accordingly. In addition, clocks may enable classification of cascades future sizes~\protect\cite{cheng2014can} and the type of non-text information (e.g. video and photos) without deep content analysis, but instead by relying on similarities to other annotated content. The latter may unlock applications such as automated hash-tagging for improved social media stream access. Given knowledge of community-specific clocks, one can also impute missing activations~\protect\cite{zong2012inferring,sadikov2011correcting}, which in turn can be used for better influence maximization in the face of a deadline 
(e.g. elections or ticket sales), and the improved quality of simulations of diffusion processes.

Our contributions in this work are as follows:\\
\noindent 1. We introduce the novel concept of heterogeneous network clocks: optimal temporal aggregations for observed information diffusion in social media. We show that while the single network clock (\oc) problem can be solved polynomially, the multiple network clocks (\koc) problem is NP-hard.\\
\noindent 2. We provide an exact solution for \oc that, while polynomial, is impractical for non-trivial instances. We also propose fast solutions for both \oc and \koc, capable of processing instances with millions of nodes and cascade activations.\\
\noindent 3. We demonstrate the utility of the network clocks for multiple applications, boosting the performance of existing algorithms for cascade size prediction (increase in $F_1$-measure by 8\%) and cascade completion (improving recall by 0.2 or more) in both synthetic and real-world data.  

%% file: fig/example.tikz
\begin{tikzpicture}[auto, thick]

\foreach \place/\x in {{(-2,1)/0}, {(-2,2)/1},{(-0.5,2)/3},
    {(-0.5,1)/4}, {(-2,0)/2}, {(-0.5,0)/5}, {(-2,-1)/6}}
  \node[cblue] (a\x) at \place {\x};
  \path[thick] (a0) edge (a1);
  \path[thick] (a0) edge (a2);
  \path[thick] (a0) edge (a3);
  \path[thick] (a1) edge (a3);
  \path[thick] (a0) edge (a2);
  \path[thick] (a2) edge (a5);
  \path[thick] (a2) edge (a6);
  \path[thick] (a3) edge (a4);
  \path[thick] (a4) edge (a5);
  \path[thick] (a5) edge (a6);

\node at (5,3){1};
\node at (6,3){2};
\node at (7,3){3};
\node at (8,3){4};
\node at (9,3){5};
\node at (10,3){6};

\node[] at (2.5,3){$X^{(1)}$};
\node[] at (2.5,2.5){$\Delta_{1}=\intv{1}{1}\intv{2}{5}\intv{6}{6}$};
\node[] at (2.5,2){$LL(\Delta_{1})=-13.9$};

\filldraw[fill=red!8!white, draw=white] (4.6,1.6) rectangle (5.3,2.8);
\draw[red,thick,<->] (4.5,1.6) --  (5.4,1.6);
\filldraw[fill=red!8!white, draw=white] (5.6,1.6) rectangle (9.3,2.8);
\draw[red,thick,<->] (5.5,1.6) --  (9.4,1.6);
\filldraw[fill=red!8!white, draw=white] (9.6,1.6) rectangle (10.3,2.8);
\draw[red,thick,<->] (9.5,1.6) --  (10.5,1.6);

\node[cred] (c16) at (5,2.5){6};
\node[cred] (c12) at (6,2){2};
\node[cred] (c15) at (9,2.5){5};
\node[cred] (c10) at (10,2){0};

\draw [->] (c16) edge (c12) (c16) edge (c15)  (c12) edge (c15) (c12) edge (c10)  ;

\node[] at (2.5,1){$X^{(2)}$};
\node[] at (2.5,0.5){$\Delta_{2}=\intv{1}{2}\intv{3}{4}\intv{5}{6}$};
\node[] at (2.5,0){$LL(\Delta_{1})=-16.3$};

\filldraw[fill=green!8!white, draw=white] (4.6,-0.5) rectangle (6.3,1.4);
\draw[green,thick,<->] (4.5,-0.5) --  (6.4,-0.5);
\filldraw[fill=green!8!white, draw=white] (6.6,-0.5) rectangle (8.3,1.4);
\draw[green,thick,<->] (6.5,-0.5) --  (8.4,-0.5);
\filldraw[fill=green!8!white, draw=white] (8.6,-0.5) rectangle (10.3,1.4);
\draw[green,thick,<->] (8.5,-0.5) --  (10.5,-0.5);

\node[cgreen] (c20) at (6,1){0};
\node[cgreen] (c21) at (7,0.6){1};
\node[cgreen] (c22) at (7,0){2};
\node[cgreen] (c23) at (8,1){3};
\node[cgreen] (c24) at (9,0.5){4};

\draw [->] (c20) edge (c21) (c20) edge (c22)  (c20) edge (c23) (c21) edge (c23) (c23) edge (c24)  ;

\node[] at (2.5,-1){$\Delta_{1,2}=\intv{1}{1}\intv{2}{2}\intv{3}{6}$};
\node[] at (2.5,-1.5){$LL(\Delta_{1,2})=-34.6$};

\filldraw[fill=black!8!white, draw=white] (4.6,-1.5) rectangle (5.3,-0.7);
\draw[black,thick,<->] (4.5,-1.5) --  (5.4,-1.5); 
\filldraw[fill=black!8!white, draw=white] (5.6,-1.5) rectangle (6.3,-0.7);
\draw[black,thick,<->] (5.5,-1.5) --  (6.4,-1.5);
\filldraw[fill=black!8!white, draw=white] (6.6,-1.5) rectangle (10.3,-0.7);
\draw[black,thick,<->] (6.5,-1.5) --  (10.5,-1.5);


\end{tikzpicture}

%% file: tex/2-prelim.tex
\vsa
\section{Problem formulation and complexity}
\vsa
\subsection{Preliminaries}
\vsa
\eat{
\begin{table*}[t]
    \centering
    \begin{tabular}{l|l}
        {\bf Symbol} & {\bf Definition}  \\ \hline
        $G(V,E)$                      & A directed unweighted graph with vertices $V$ and edges $E$ \\ \hline
        $N_u$                        & In-neighbors of vertex $u$ in $G$ \\ \hline
        $\mathbb{T} = \{0 \ldots T\}$ & Original timeline of discrete time points \\ \hline
        $X^{(k)}=\{(u,t)\}$                 & A cascade (contagion): activation node-time pairs $(u,t),u\in V, t\in\mathbb{T}$ \\ \hline
        $A_t^{(k)}$                   & Activated nodes at time $t$, i.e. $\{u|(u,t)\in X_k\}$ \\ \hline
        $N_{u,t}^{(k)}=A_{t-1}^{(k)} \cap N_u$ & Contagious neighbors of $u$ at time $t$ \\ \hline
        $p_c,p_e$                     & Neighbor $p_c$ and external $p_e$ activation prob. in the IC model \\ \hline
        $L_{u,t}^{(k)}=1 - (1 - p_e)(1 - p_c)^{|N_{u,t}^{(k)}|}$ & Likelihood of observed activation of $u$ at $t$  \\ \hline
    \end{tabular}
    \caption{Notation}
    \label{tab:my_label}
\end{table*}
}

We denote a static directed network as $G(V,E)$ with nodes $V$ and edges $E \subseteq V \times V$. We refer to individual instances of diffusion processes over time through this network as \emph{cascades}. Thus, a cascade $X_i=\{(u,t)\}$ is a set of node-time pairs, each representing the activation of node $u$ at time $t$, where the set of all observed cascades in a data set is $\mathbb{X} = \{X_1, \ldots, X_m\}$. We assume a discrete timeline where $1$ and $T$ are the times of the first and last activation events in $\mathbb{X}$, respectively. We also assume that each cascade activates a vertex at most once. 

Among the multiple models for cascade processes, we focus on the well-studied independent cascade (IC) model initially introduced by Goldenberg et al.~\protect\cite{goldenberg2001using}. The key intuition behind the model is that, with respect to a given cascade, each node may be in only one of three states: inactive, active, or previously activated. The critical notion in this model is that nodes can only influence neighbors for a single time period after activation. In other words, a node that is activated at time $t$ increases the likelihood that its neighbors activate at time $t+1$, but not at $t+2$ or any later time.

The IC model has two key parameters: the probability $p_n$ that a newly activated node activates any network neighbor, and the probability $p_e$ that an inactive node \say{spontaneously} activates due to external factors. Note that the above assumes homogeneous activation probabilities for all nodes and neighbor links. Our methods are directly applicable to the non-homogeneous IC setting where probabilities may differ across nodes and edges, but we do not explore that setting here. We assume that $p_e$ is much smaller than $p_n$; otherwise the network structure becomes much less important. As suggested by the name of the model, all activation events are presumed to be independent of each other.

The set of incoming neighbors of $u$ is denoted by $N(u)=\{v\ |\ (v,u) \in E\}$; these are the nodes that can potentially influence $u$. The set of nodes activated as part of a cascade $X$ at time $t$ is denoted by $A(X,t)$. The activation time of node $v$ within cascade $X$ is denoted as $X(v)$. Finally, the set of newly active neighbors of $v$ at time $t$ are denoted as $C(X,v,t) = A(X,t-1)\ \cap\ N(v)$, i.e. those are the nodes that can potentially activate $v$ at time $t$.

Given the above definitions, we can quantify the {\em a priori\em} likelihood that node $v$ is part of cascade $X$ at time $t$ as follows: \vsa 
\begin{equation}
  L(v,X,t) = \begin{cases}
    1 \text{, if $X(v) < t$ i.e. $v$ is already in $X$} \\
    1 - (1 - p_e)(1 - p_n)^{|C(X,v,t)|} \text{, o/w.}
\end{cases}
\end{equation}
This likelihood reflects the possibility of a node being activated by at least one contagious neighbor or by sources external to the network. In order to characterize the likelihood of all independent observations, we use the common practice of working with log-likelihoods. Thus, the log-likelihood of node $v$ joining $X$ at time $t$ is $LL(X,v,t) = \log{L(X,v,t)}$. For convenience, we also define the log likelihood of $v$ not joining $X$ at time $t$ as $\overline{LL}(X,v,t) = \log{1 - L(X,v,t)}$. We can now quantify the log-likelihood of all observed cascades $\mathbb{X}$ and their activation events as: \vsa
\begin{equation}
  \label{eq:ll}
  LL(\mathbb{X}) = \sum_{X \in \mathbb{X}}\sum_{t=1}^T\big(\hspace{-8pt}
  \sum_{\hspace{2pt}u \in A(X,t)}\hspace{-8pt}LL(u,X,t) +\hspace{-8pt} \sum_{v \notin A(X,t)}\hspace{-8pt}\overline{LL}(v,X,t)
  \big) 
\end{equation}
Note that the independence assumption across cascades and activation events results in the above sum of log-likelihoods. 


%% file: tex/3-problem.tex
\vsa
\subsection{Network clocks and likelihood improvement}
\vsa

The independent cascade model assumes discrete time of events and that all nodes are participants in the diffusion process at every time step. Such assumption enforces a uniform mapping of ``wall-clock'' time to network events across the entire network, and thus, renders IC models infeasible to interpret information diffusion in real-world networks, which are often not governed by a single clock.

For example, consider an online social network user who checks her feed upon waking up, once or twice during lunch break, and several times in the evening. No uniform timeline can explain her activity sessions well - short time steps will make morning activity appear unrelated to posts that came in overnight, while long time steps may lose information about rapid responses that occur in the evening. A {\em heterogeneous\em} division of the timeline, which we term a \emph{network clock} is needed to accurately explain such activity.

A network clock $\Delta$ is a partitioning of $\mathbb{T}$ into contiguous segments of the original timeline corresponding to ``new'' discrete time steps. In effect, the events that occur within each interval of the timeline are considered to be simultaneous. In the example from Fig.~\ref{fig:ex}, the clock $\Delta_1$ for cascade $X^{(1)}$ partitions the time into intervals $\intv{1}{1},\intv{2}{5}$ and $\intv{6}{6}$, which renders the activations of nodes $2$ and $5$ simultaneous. 

Formally, a network clock $\Delta = \{\delta_1,\ldots, \delta_f\}$ is a set of adjacent non-overlapping time intervals $\delta_i = \intv{t_s}{t_e}$ that span the original timeline $\mathbb{T}$, i.e. :\vsa
$$t_s(\delta_i) \leq t_e(\delta_i), t_s(\delta_i) = t_e(\delta_{i-1}) + 1, \forall i.\vsa$$


The two extreme clocks are the \emph{maximum aggregation} (i.e. all original timestamps are simultaneous): $\Delta_{max} = \{ \intv{1}{T}\}$, and the \emph{minimum aggregation}, corresponding to the original timeline: $\Delta_{min} = \{\intv{1}{1}, \ldots, \intv{T}{T}\}$. The total number of possible clocks is $2^{T-1}$, since we can either split or not between any two consecutive points in $\intv{1}{T}$.

For a fixed clock $\Delta$, we treat $\delta_i$ as individual consecutive time steps (indexed by $i$); however, as necessary we also use $\delta_\Delta(t)$ to denote the unique interval $\delta \in \Delta$ containing the original time $t$. Given a clock $\Delta$, the set of active nodes at period $\delta$ is all activations that happened during $\delta$, denoted $A(X,\delta) = \bigcup_{t = t_s(\delta)}^{t_e(\delta)}A(X,t)$; other temporal entities are similarly mapped from the original times $t$ to periods $\delta$. We can now redefine the clock likelihood of an observed activation during interval $\delta_i$ following interval $\delta_{i-1}$ as:\vsa
\begin{equation}
  \label{eq:intlike}
  L(v,X,\delta_i\leftarrow\delta_{i-1}) = \begin{cases}
    1 \text{, when $X(v) < t_s(\delta_i)$} \\
    1 - (1 - p_e)(1 - p_n)^{|C(X,v,\delta_i)|} \text{, o/w.}
\end{cases}
\end{equation}

\noindent Note that the set of contagious neighbors of $v$ at time $\delta_i$ are determined by the preceding interval: $C(X,v,\delta_i) = N(v) \cap A(X,\delta_{i-1})$. As a result, different clocks will result in different likelihoods of observed events. We denote the log-likelihood of all events given a clock as $LL(\mathbb{X}|\Delta)$, defined analogously to Eq.~\ref{eq:ll}, where the sum over $t$ and selection of active nodes is replaced by a sum over clock intervals.
\eat{
\begin{equation*}
  LL(\mathbb{X}|\Delta) = \sum_{X \in \mathbb{X}}\sum_{\delta \in \Delta}\big(\hspace{-8pt}
  \sum_{\hspace{2pt}u \in A(X,\delta)}\hspace{-8pt}LL(X,u,\delta) +\hspace{-8pt} \sum_{v \notin A(X,\delta)}\hspace{-8pt}\overline{LL}(X,v,\delta)
  \big) 
\end{equation*}
}

Our goal is to find a clock (if any) that provides a better explanation of observed cascades $\mathbb{X}$ than the timeline of $\Delta_{max}$, a default clock that groups all events into a single step. A clock improvement $I(\mathbb{X}|\Delta)$ is the increase in likelihood compared to that of the default clock:\vsa
$$I(\mathbb{X}|\Delta)=LL(\mathbb{X}|\Delta) - LL(\mathbb{X}|\Delta_{max}).\vsa$$


In our example in Fig.~\ref{fig:ex} the default clock $\Delta_{max}$ will render the activation of all nodes spontaneous (i.e. activated from external sources) due to the lack of preceding time steps, and thus, of contagious neighbors. 
The log-likelihood of cascade $X^{(1)}$ (in red) under the default clock $\Delta_{max}$ is $L(X^{(1)}|\Delta_{min})=-23.5$, while the same quantity for clock $\Delta_1$ is $L(X^{(1)}|\Delta_{1})=-13.9$, hence the improvement of $\Delta_{1}$ for cascade $X^{(1)}$ is $I(X^{(1)}|\Delta_1)=9.6$. Similarly, the improvement of the clock $\Delta_{1,2}$ when both cascades are considered simultaneously can be calculated as $I(\{X^{(1)},X^{(1)}\}|\Delta_{1,2})=-34.6-(-40)=5.4$.

\vsa
\subsection{Optimal network clock problems: \oc and \koc}
\vsa

We are now ready to define our clock optimization problems: the \emph{Optimal clock problem (\oc)} which seeks to identify a temporal clock that best explains all observed cascades; and the \emph{k-Optimal Clock problem (\koc)} which maps nodes in the network to one of $k\leq|V|$ clocks corresponding to different temporal behaviors of network users. We show that the former can be solved in polynomial time, while the latter is NP-hard.
\vsb
\begin{definition}\emph{\bf [Optimal clock problem (\oc)]}
\label{prob:o1c}
Given a network $G$, a set of cascade observations $\mathbb{X}$, and probabilities $p_e$ and $p_n$, find the clock $\Delta^*=\argmax_{\Delta}I(\mathbb{X}|\Delta)$.\vsa
\end{definition}

Problem \oc asks for the single partition of time that best captures the temporal properties of all observed cascades $\mathbb{X}$ given the IC model. The problem is most applicable in cases where either all of the network nodes have similar temporal habits of social network interactions (e.g. 9am-5pm office workers), or when cascades progress in a similar fashion. In such cases it is desirable to detect this shared temporal behavior across the whole network.

Unlike a homogeneous clock, the optimal solution of \oc will adaptively segment time in order to increase the likelihood of observed data. Thus, periods of bursty activity may be partitioned into more steps, while those of relatively low activity combined into a single step. The number of potential clocks is $2^{T-1}$; nevertheless, \oc admits a polynomial-time exact solution based on dynamic programming discussed in detail in the following section. 

The assumption of a single temporal clock shared among all network nodes might be too stringent considering the size, scope and heterogeneity of real-world OSN users. Beyond 9am-5pm office worker activity patterns, we might observe those of high-school or college students, stay-at-home parents or seniors and in addition all these diverse demographics may reside in different time zones, while still participating in the same global OSN. To account for such diverse population of users we also consider a multiple clock version of the problem, where each node ``follows'' a single clock. 

One can think of a clock as how the events are temporally perceived from the viewpoint of a node. Instead of imposing the same clock on all nodes, we can model different behavioral patterns as multiple clocks that individual nodes follow. Given a set of clocks $\mathbb{D}=\{\Delta_1, \Delta_2..., \Delta_k\}$, we define their improvement as:\vsa
\begin{equation}
\label{eq:mci}
I(\mathbb{X}|\mathbb{D})=\sum_{v\in V} \max_{\Delta \in \mathbb{D}} \Big(LL(\mathbb{X}(v)|\Delta) - LL(\mathbb{X}(v)|\Delta_{min})\Big),\vsb
\end{equation}

\noindent where $\mathbb{X}(v)$ denotes the activations of node $v$ in all cascades in $\mathbb{X}$. Note that  $I(\mathbb{X}|\mathbb{D})$ can be evaluated in polynomial time if the set of clocks $\mathbb{D}$ is predefined by calculating the likelihood of each node $v$'s activations according to each clock. 

In our example in Fig.~\ref{fig:ex} clock $\Delta_{1,2}$ is the best single clock when considering both cascades $\mathbb{X}=\{X^{(1)},X^{(2)}\}$; however, under this clock, the influence from $6\rightarrow 5$ in $X^{(1)}$
is discounted since $6$ is activated in the first interval of $\Delta_{1,2}$, while $5$ is activated in its third interval. Similarly, the potential influences from $1\rightarrow 3$ and $3 \rightarrow 4$ in $X^{(2)}$ are discarded as all of those activations are rendered simultaneous according to $\Delta_{1,2}$. Consider, however, a second alternative clock for these cascades $\Delta=\intv{1}{3}\intv{4}{4}\intv{5}{6}$. The improvement of both clocks will exceed that of $\Delta_{1,2}$ alone, i.e.  $I(\mathbb{X}|\Delta_{1,2},\Delta)>I(\mathbb{X}|\Delta_{1,2})$, since the likelihoods of the activations of nodes $3$ and $4$ will increase by being reassigned to $\Delta$, thus considering all of their potential network activations.\vsa

\begin{definition}\emph{\bf [k-Optimal clock problem (\koc)]}
Given a network $G$, a set of cascade observations $\mathbb{X}$, probabilities $p_e$ and $p_n$, and a positive integer $k>0$ find the set of clocks $\mathbb{D}^*=\argmax_{\mathbb{D}}I(\mathbb{X}|\mathbb{D}), $ s.t. $|\mathbb{D}|\leq k$.\vsa
\end{definition}

For values of $k=1$ or $k\geq |V|$, \koc can be solved in polynomial time. $k=1$ is equivalent to the \oc problem and can be solved using dynamic programming as discussed in the next section, and when the budget of clocks exceeds the number of graph nodes, $k\geq |V|$, we can find the single best clock that optimizes the improvement of a given node's activations and then assign this clock to the node. However, for general $k<|V|$ we show that the problem is NP-hard by a reduction from \emph{Graph coloring} to a decision version of \koc. \vsb

\begin{theorem}
kOC is NP-hard.\vsa
\end{theorem}

\begin{proof}
We show a reduction from $k$-Graph Vertex Coloring (kVC) to a decision version of \koc. Given a graph, kVC asks for a mapping of each node to at most $k$ colors such that no two colors are adjacent. Let \koc-B be a decision version of \koc asking if there exist $k$ clocks s.t. $I(\mathbb{X}|\mathbb{D})>B, B>0$. To construct the reduction we will need to show how we obtain an instance \koc from an instance of kVC and also set the values of $B$ and $k$.

Given a kVC instance $G(V,E)$ we create an instance of \koc involving a graph $G'(V',E')$ with $|V| + 2|E|$ nodes, $4|E|$ edges and $4|E|$ cascades using the gadget in Fig.~\ref{fig:npgadget}. For every kVC edge $(u,v)$, we create four cascades involving two activations each. The later activations of each of the cascades involve nodes corresponding to the nodes in the kVC instance, while the earlier activations involve unique other nodes. For example, the gadget corresponding to edge $(0,1)$ in Fig.~\ref{fig:npgadget} involves the four cascades spanning times $[1-4]$. Two cascades ``end''  with the equivalent node to $0$ in kOC: $(a_{01},t_1) \rightarrow (0,t_2)$ and $(b_{01},t_3)\rightarrow (0,t_4)$; and another two with the equivalent of node $1$: $(c_{01},t_1) \rightarrow (1,t_4)$ and $(d_{01},t_1)\rightarrow (1,t_4)$. Note that it is not possible to devise a single clock that renders all $4$ latter activations ($(0,t_2)$,$(0,t_4)$,and the two $(1,t_4)$) of the cascades non-spontaneous, i.e. there is no set of splits in time that will ``cut'' all $4$ cascade activation edges exactly once. 

Applying this transformation for each kVC edge we obtain an instance of $4|E|$ cascades, where a pair of cascades coincides in time only if they are within the same gadget. We set the log-likelihood requirement as $B=4|E|\log p_n$, where we set $p_e$ and $p_n$ such that $\log p_e < 2|E|\log p_n$. In other words, we require $k$ clocks that ensure that all latter gadget activations are non-spontaneous.

If there exists a kVC solution, one can create a corresponding \koc solution by assigning one clock per color ensuring that all activations of the corresponding nodes of that color are non-simultaneous in each gadget. This is possible since nodes of the same color will never appear in the same gadget. Consider, as an example, the set of corresponding $4$ clocks for a $4$-coloring of the kVC instance in Fig.~\ref{fig:npgadget}. Alternatively, assume there exists a \koc solution in which every cascade activation edge is cut exactly once by the clock assigned to the later node. Then all activations of a given end-node will be rendered network-induced by at least one clock and we can assign colors in the kVC instance based on this clock assignment ensuring that colors will not be adjacent. \vsa
\end{proof}
\begin{figure}[t]\vsa
\centering
 \resizebox{0.48\textwidth}{!}{%
 \input{fig/np-proof.tikz}\vsb
}
\caption{ \footnotesize Gadget used in the reduction from Vertex Coloring (kVC) to \koc.}\vsa
\label{fig:npgadget}
\end{figure}
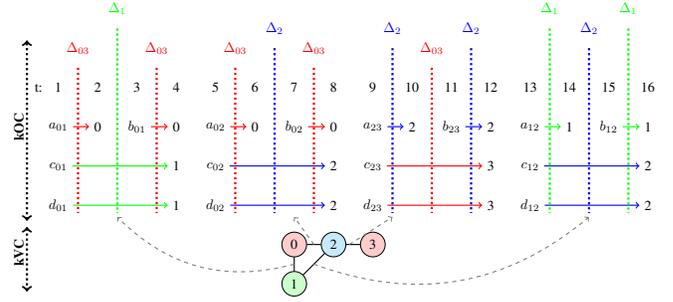

%% file: fig/np-proof.tikz
\begin{tikzpicture}[auto, thick]

\node[rotate=90] at (1,-3.2){\bf kVC};
\draw[ultra thick,dotted,<->] (1.2,-2.55) --  (1.2,-4.2);

\node[rotate=90] at (1,0){\bf kOC};
\draw[ultra thick,dotted,<->] (1.2,-2.4) --  (1.2,2.2);

  \node[cred] (a0) at (8,-3){0};
  \node[cgreen] (a1) at (8,-4){1};
  \node[cblue] (a2) at (9,-3){2};
  \node[cred] (a3) at (10,-3){3};
  \path[thick] (a0) edge (a1);
  \path[thick] (a1) edge (a2);
  \path[thick] (a0) edge (a2);
  \path[thick] (a2) edge (a3);

     \coordinate (01e) at (8,-3.5) ;
     \coordinate (01g) at (3.5,-2.3) ;
     \draw[->,gray,dashed] (01e) to [bend left] (01g) ;

     \coordinate (12e) at (8.5,-3.5) ;
     \coordinate (12g) at (15.5,-2.3) ;
     \draw[->,gray,dashed] (12e) to [bend right] (12g) ;

     \coordinate (02e) at (8.5,-3) ;
     \coordinate (02g) at (8,-2.3) ;
     \draw[->,gray,dashed] (02e) to (02g) ;

     \coordinate (23e) at (9.4,-3) ;
     \coordinate (23g) at (10.5,-2.3) ;
     \draw[->,gray,dashed] (23e) to (23g) ;

\node at (1.5,1){t:};
\node at (2,1){1};
\node at (3,1){2};
\node at (4,1){3};
\node at (5,1){4};
\node at (6,1){5};
\node at (7,1){6};
\node at (8,1){7};
\node at (9,1){8};
\node at (10,1){9};
\node at (11,1){10};
\node at (12,1){11};
\node at (13,1){12};
\node at (14,1){13};
\node at (15,1){14};
\node at (16,1){15};
\node at (17,1){16};


\node (c11) at (2,0){$a_{01}$};
\node (c12) at (3,0){0};
\node (c13) at (4,0){$b_{01}$};
\node (c14) at (5,0){0};
\node (c15) at (2,-1){$c_{01}$};
\node (c16) at (5,-1){1};
\node (c17) at (2,-2){$d_{01}$};
\node (c18) at (5,-2){1};

\draw [red,->] (c11) edge (c12) (c13) edge (c14);
\draw [green,->] (c15) edge (c16) (c17) edge (c18);

\node (c21) at (6,0){$a_{02}$};
\node (c22) at (7,0){0};
\node (c23) at (8,0){$b_{02}$};
\node (c24) at (9,0){0};
\node (c25) at (6,-1){$c_{02}$};
\node (c26) at (9,-1){2};
\node (c27) at (6,-2){$d_{02}$};
\node (c28) at (9,-2){2};

\draw [red,->] (c21) edge (c22) (c23) edge (c24);
\draw [blue,->] (c25) edge (c26) (c27) edge (c28);

\node (c31) at (10,0){$a_{23}$};
\node (c32) at (11,0){2};
\node (c33) at (12,0){$b_{23}$};
\node (c34) at (13,0){2};
\node (c35) at (10,-1){$c_{23}$};
\node (c36) at (13,-1){3};
\node (c37) at (10,-2){$d_{23}$};
\node (c38) at (13,-2){3};

\draw [blue,->] (c31) edge (c32) (c33) edge (c34);
\draw [red,->] (c35) edge (c36) (c37) edge (c38);

\node (c41) at (14,0){$a_{12}$};
\node (c42) at (15,0){1};
\node (c43) at (16,0){$b_{12}$};
\node (c44) at (17,0){1};
\node (c45) at (14,-1){$c_{12}$};
\node (c46) at (17,-1){2};
\node (c47) at (14,-2){$d_{12}$};
\node (c48) at (17,-2){2};

\draw [green,->] (c41) edge (c42) (c43) edge (c44);
\draw [blue,->](c45) edge (c46) (c47) edge (c48);

\node[red] at (2.5,2){$\Delta_{03}$};
\coordinate (C11) at (2.5,1.5) ;
\coordinate (C12) at (2.5,-2.2) ;
\draw[dotted,red,ultra thick] (C11) to (C12) ;

\node[red] at (4.5,2){$\Delta_{03}$};
\coordinate (C11) at (4.5,1.5) ;
\coordinate (C12) at (4.5,-2.2) ;
\draw[dotted,red,ultra thick] (C11) to (C12) ;

\node[red] at (6.5,2){$\Delta_{03}$};
\coordinate (C11) at (6.5,1.5) ;
\coordinate (C12) at (6.5,-2.2) ;
\draw[dotted,red,ultra thick] (C11) to (C12) ;

\node[red] at (8.5,2){$\Delta_{03}$};
\coordinate (C11) at (8.5,1.5) ;
\coordinate (C12) at (8.5,-2.2) ;
\draw[dotted,red,ultra thick] (C11) to (C12) ;

\node[red] at (11.5,2){$\Delta_{03}$};
\coordinate (C11) at (11.5,1.5) ;
\coordinate (C12) at (11.5,-2.2) ;
\draw[dotted,red,ultra thick] (C11) to (C12) ;

\node[green] at (3.5,3){$\Delta_{1}$};
\coordinate (C11) at (3.5,2.5) ;
\coordinate (C12) at (3.5,-2.2) ;
\draw[dotted,green,ultra thick] (C11) to (C12) ;

\node[green] at (14.5,3){$\Delta_{1}$};
\coordinate (C11) at (14.5,2.5) ;
\coordinate (C12) at (14.5,-2.2) ;
\draw[dotted,green,ultra thick] (C11) to (C12) ;

\node[green] at (16.5,3){$\Delta_{1}$};
\coordinate (C11) at (16.5,2.5) ;
\coordinate (C12) at (16.5,-2.2) ;
\draw[dotted,green,ultra thick] (C11) to (C12) ;

\node[blue] at (7.5,2.5){$\Delta_{2}$};
\coordinate (C11) at (7.5,2) ;
\coordinate (C12) at (7.5,-2.2) ;
\draw[dotted,blue,ultra thick] (C11) to (C12) ;

\node[blue] at (10.5,2.5){$\Delta_{2}$};
\coordinate (C11) at (10.5,2) ;
\coordinate (C12) at (10.5,-2.2) ;
\draw[dotted,blue,ultra thick] (C11) to (C12) ;

\node[blue] at (12.5,2.5){$\Delta_{2}$};
\coordinate (C11) at (12.5,2) ;
\coordinate (C12) at (12.5,-2.2) ;
\draw[dotted,blue,ultra thick] (C11) to (C12) ;

\node[blue] at (15.5,2.5){$\Delta_{2}$};
\coordinate (C11) at (15.5,2) ;
\coordinate (C12) at (15.5,-2.2) ;
\draw[dotted,blue,ultra thick] (C11) to (C12) ;

\end{tikzpicture}

%% file: tex/4-sol.tex
\vsa
\section{Algorithms for optimal clock detection}
\vsa

Next we present our solutions for \oc and \koc. We first demonstrate that an exact solution for \oc is possible in polynomial time employing dynamic programming (\dmp). Its complexity, however, makes it challenging to apply to real-world instances; hence we propose a fast and accurate greedy (\gr) alternative. Although \koc is NP-hard, we show that the improvement function is \emph{monotonic} and \emph{submodular} with respect to the set of solution clocks, leading to greedy generalizations \grdp and \grsq of our single clock solutions, the former featuring an approximation guarantee. Tbl.~\ref{tbl:algs} summarizes the algorithms and results discussed in this section.
\vsc

\begin{table}
\centering
\footnotesize
\begin{tabular}{|c|c|c|c|c|} \hline
    {\bf Algorithm}& {\bf Problem}& {\bf Complexity}& {\bf Asymptotic Time}& {\bf Guarantee}\\ \hline 
    \dmp & $\oc$&Poly &  $O(|\mathbb{X}|^4)$& {\em Exact\em}\\ \hline
    \gr & $\oc$& Poly&$O(m|\mathbb{X}|\log(|\mathbb{X}|))$& \\ \hline
    \grdp & \koc& NP-hard & $O(k|\mathbb{X}|^4)$ & $1-1/e$\\ \hline
    \grsq & \koc& NP-hard & $O(km|\mathbb{X}|\log(|\mathbb{X}|))$& \\ \hline
\end{tabular}
\caption{Algorithms, complexity and approximation guarantees.}
\label{tbl:algs}
\end{table}

\subsection{Single-clock approaches}
\vsa

While there are $O(2^{T-1})$ possible clocks for a timeline of $T$ original steps, we show that we can compute the optimal one by considering $O(T^2)$ intervals and employing dynamic programming. Let $S:\{\delta^e_s=\intv{s}{e}\}$ be the set of all possible sub-intervals of the timeline $[1,T]$, where $1\leq s\leq e \leq T$ denote the start and end of the interval. Note that there is a quadratic number of possible intervals $|S|=T(T+1)/2$. Let also $\delta^e_{\leftarrow}$ denote the set of all intervals ending at $e$ and starting at a time no later than $e$. Let the improvement of $v$ in cascade $X$ over interval $\delta^e_s$ given a preceding interval $\delta^{s-1}_b$ be:
\begin{multline}
     I(v,X,\delta^e_s\leftarrow \delta^{s-1}_b)= \\
\threepartdef
{LL(v,X,\delta^e_s\leftarrow \delta^{s-1}_b)}      {X(v)>e}
{LL(v,X,\delta^e_s\leftarrow \delta^{s-1}_b)-\log p_e}      {X(v)\in\delta_s^e}
{0} {X(v)<s,}
\label{eq:intimpr}
\end{multline}
where $LL(v,X,\delta^e_s\leftarrow \delta^{s-1}_b)$ is the interval log-likelihood of $v$ defined as in Eq.~\ref{eq:intlike} with the set of active neighbors of $v$ $C(X,v,\delta^e_s\leftarrow \delta^{s-1}_b)=N(v)\cap A(X,\delta^{s-1}_b)$. Note that in the above definition we subtract the log-likelihood of the default clock $\Delta_{min}$ in case the activation of $v$ is within the interval $X(v)\in\delta_s^e$ to quantify the improvement due to $v$'s activation. The total improvement of the interval transition $\delta^e_s$ is defined as the sum of all node improvements in all cascades: \vsa
\begin{equation}
    I(\delta^e_s\leftarrow \delta^{s-1}_b,\mathbb{X}) = \sum_{X \in \mathbb{X}}\sum_{v \in V} I(v,X,\delta^e_s\leftarrow \delta^{s-1}_b).\vsa
    \label{eq:totalintimpr}
\end{equation} 

We will denote with $I^*(\delta^e_s,\mathbb{X})$ the best improvement of log-likelihood up-to and including interval $\delta^e_s$, over all sequences of preceding intervals. In other words, $I^*(\delta^e_s,\mathbb{X})$ is the optimal clock up to time $e$, such that the last interval in the clock is $\delta^e_s$. We can recursively define this quantity as: \vsb
\begin{equation}
    I^*(\delta^e_s,\mathbb{X})=\max_{\delta \in \delta^{s-1}_{\leftarrow}} \Big[  I^*(\delta,\mathbb{X}) + I(\delta^e_s\leftarrow \delta,\mathbb{X})\Big].\vsa
    \label{eq:optintimpr}
\end{equation}

The maximum in the definition is taken over the set of all intervals ending immediately before $\delta^e_s$: $\delta^{s-1}_{\leftarrow}$. Based on the above definition, the improvement of the best clock is the maximum over the best improvements of intervals ending at $T$, i.e. $I(\mathbb{X}|\Delta^*)=\argmax_{\delta \in \delta^T_{\leftarrow}}I^*(\delta,\mathbb{X})$. This equivalence gives us the blueprint for a dynamic programming (DP) solution that progressively computes the improvement of all $T(T+1)/2$ intervals and uses back-tracking to retrieve the optimal clock.  

\begin{algorithm}[t]
\footnotesize
\caption{\dmp}\label{alg:singleDP}
\begin{algorithmic}[1]
\Require{Graph $G$, set of cascades $\mathbb{X}$}
\Ensure{Clock $\Delta^*$ solution to OC}
\For{$e=1\dots T$}
  \State $I^*(\delta^e_1,\mathbb{X}) = I(\delta^e_1\leftarrow \emptyset,\mathbb{X})=0$  \emph{// Init. intervals starting at t=1}
\EndFor
\For{$s=2\dots T$}
  \For {$e=s\dots T$}
    \State $I^*(\delta^e_s,\mathbb{X})=\max_{\delta \in \delta^{s-1}_{\leftarrow}} \big[  I^*(\delta,\mathbb{X}) + I(\delta^e_s\leftarrow \delta,\mathbb{X})\big]$
    \State $bt(\delta^e_s)=\delta^{s-1}_b$, where $\delta^{s-1}_b$ is the maximizer in $I^*(\delta^e_s,\mathbb{X})$
  \EndFor
\EndFor
\State $\delta^T_b=\argmax_{\delta \in \delta^T_{\leftarrow}}I^*(\delta,\mathbb{X})$ // Last interval of $\Delta^*$
\State Reconstruct $\Delta^*$ by recursively following $bt()$ starting from $bt(\delta^T_b)$
\State \Return{$\Delta^*$}
\end{algorithmic}
\end{algorithm}

\begin{algorithm}
\footnotesize
\caption{\gr}\label{alg:singleGreedy}
\begin{algorithmic}[1]
\Require{Graph $G$, set of cascades $\mathbb{X}$}
\Ensure{Clock $\Delta$}
\State $\Delta = \{\intv{1}{T}\}$
\State $\mathbb{E}_a=\{(u,v,X)|(u,v)\in E,X(u)<X(v)\}$ \emph{// Active edges}
\Repeat
  \State $\mathbb{C}=\emptyset$ \emph{// Set of independent cut positions}
  \For{$t$ in increasing activation times of all $(u,t',X)\in\mathbb{X}$}
    \State Add active edges $(u,v,X)$ to $\mathbb{E}_a(t)$ s.t. $X(v)=t$
    \State Compute $I(\mathbb{E}_a(t))$
    \IIf{$I(\mathbb{E}_a(t))>0$} add-or-drop$((t,\mathbb{E}_a(t))\rightarrow \mathbb{C})$ \EndIIf
    \State Remove active edges $(u,v,X)$ from $\mathbb{E}_a(t)$ s.t. $X(u)=t$ 
  \EndFor
  \IIf{$\mathbb{C}\neq \emptyset$} Add all cuts $\mathbb{C}$ to $\Delta$ \EndIIf
\Until $\mathbb{C}==\emptyset$
\State \Return $\Delta$
\end{algorithmic}
\end{algorithm}

Alg.~\ref{alg:singleDP} presents the steps of the \dmp algorithm based on Eq.~\ref{eq:optintimpr} producing the optimal solution of \oc.
Steps 1-3 compute the best improvement for the base cases of the recursion, i.e. all intervals staring at $t=1$ which do not have preceding intervals. Since no nodes are active before this time, all activations are considered spontaneous, and thus the individual node improvement scores and the cumulative $I(\delta^e_1\leftarrow \emptyset,\mathbb{X})$ are all zero. The nested loops in Steps 4-9 incrementally compute the best improvement (Step 6) of intervals of increasing starting position and record the preceding interval that leads to the best improvement to be later used by backtracking (Step 7). Note that at the time at which interval $\delta^e_s$ is considered, the best improvement of all immediately preceding intervals $I^*(\delta^{s-1}_\leftarrow,\mathbb{X})$ is already available, thus allowing the recursive computation of $I^*(\delta^e_s,\mathbb{X})$ in Step 6. Finally, the interval of best improvement is selected among those that end at $T$ and the optimal clock is reconstructed by backtracking from this interval using $bt()$ (Steps 10-11). 

\noindent{\it \dmp Complexity: } Alg.~\ref{alg:singleDP} runs in time $O(|\mathbb{X}|T^3)$, where $|\mathbb{X}|$ denotes the number of all activations in all cascades. This is due to the nested loops (Steps 4-9) which consider all $O(T^2)$ intervals and the maximization in line $6$ which iterate over $O(T)$ preceding intervals of the current $\delta^e_s$. The complexity does not depend on the full graph $G(V,E)$, since only nodes participating in the cascades and downstream neighbors need to be considered (the likelihood of non-active nodes without active neighbors can be computed in constant time for any interval). In addition, when there are time steps without activations one can ``drop'' such intervals in a pre-processing step, effectively letting activations define the original time points. In this case the complexity becomes effectively $O(|\mathbb{X}|^4)$. 

While \dmp solves \oc exactly in polynomial time, its complexity is prohibitive for large real-world instances. \gr is a top-down greedy approach that scales almost linearly, adding one partition at a time by splitting existing intervals, starting from the default clock $\Delta_{max}$. The greedy aspect is in the selection of the next split based on maximum improvement of the score. 

The steps of \gr are outlined in Alg.~\ref{alg:singleGreedy}. It initializes the solution with the default aggregate clock $\Delta_{max}$ (Step 1) and computes all active edges $\mathbb{E}_a$ within cascades such that the source and destination are active preserving the temporal order (Step 2). It then iteratively tries to add groups of independent cuts $\mathbb{C}$ that do not intersect common active edges in Steps 3-12. $\mathbb{C}$ is empty at the beginning of each iteration (Step 4). The inner loop (Steps 5-10) is essentially a modified \emph{sweep line} algorithm in 1D~\protect\cite{shamos1976geometric}, where the segments are defined by active edges in $\mathbb{E}_a$. Going from small to large times $t$ of all activations in $\mathbb{X}$, we first add active edges that ``end'' at $t$ to $\mathbb{E}_a(t)$ (Step 6) and recompute the improvement $I(\mathbb{E}_a(t))$ based on ``current'' active edges $\mathbb{E}_a(t)$ and based on whether they have been already cut by previous sweeps (Step 7). If the improvement at time $t$ exceeds $0$, we attempt to add the cut $(t,\mathbb{E}_a(t))$ to the $\mathbb{C}$ (Step 8). A cut is added if (i) it does not have shared activations with any other cut in $\mathbb{C}$, or (ii) in the case it does, it is added if its improvement exceeds that of the unique cut it shares active edges with. This logic in \emph{add-or-drop()}(Step 8) ensures that $\mathbb{C}$ maintains the set of independent time cuts of highest improvement at every round of the outer loop. In Step 9 we remove active edges ending at $t$, since they will not be considered at $t+1$. If $\mathbb{C}$ is not empty (Step 12) we add all cuts to $\Delta$ and update their active edges to a cut state, which affects the improvement computations on Step 7 in subsequent iterations. Finally, if no new cuts were identified we terminate the loop and return the clock (Steps 12-13).

\noindent{\it \gr complexity: } Each sweep requires $O(|\mathbb{X}|\log (|E_a|))$ time, since we move through the activations and maintain the currently active edges (that enter and exit only) once the structure $E_a(t)$, where the log factor comes from maintaining them in a heap structure that support add (Step 6) and delete (Step 9). Since $|E_a|$ grows slower than $O(|\mathbb{X}|^2)$ the sweeps take $O(|\mathbb{X}|\log (|\mathbb{X}|))$. Generating all $|E_a|$ in Step 2 might in general be higher than $O(|\mathbb{X}|\log (|\mathbb{X}|))$ for very dense networks, but assuming that the average degree grows as $O(\log |V|)$ this step will not exceed in complexity a single sweep. The number of sweeps can also reach $m=O(|\mathbb{X}|)$ if we make one cut at a time on every sweep, but such instances will have lots of overlapping active edges in time. The total time is log-linear $O(m|\mathbb{X}|\log(|\mathbb{X}|))$ if we assume that not too many sweeps $m$ are performed. As demonstrated in the evaluation section, \gr scales near-linearly to large real-world instances and produces very high-quality solutions ($>90\%$ of optimal). 
\vsb
\subsection{Multi-clock approaches}
\vsa
Unlike \oc, \koc is NP-hard. We propose a greedy \grdp solution and show that it ensures a $(1-1/e)$-approximation. 

\begin{theorem}\vsa
$I(\mathbb{X}|\mathbb{D})$ is monotonic and submodular as a function of the set of clocks $\mathbb{D}$.
\end{theorem}

\begin{proof}\vsa
The monoticity is due to the maximization in the definition of the improvement (Eq.~\ref{eq:mci}): $I(\mathbb{D} \cup \Delta) \geq I(\mathbb{D})$, since in the worst case no node will ``select'' $\Delta$ rendering the improvement unchanged due to its addition.

Given two nested sets of clocks $\mathbb{D} \subseteq \mathbb{C}$ and $\Delta \notin \mathbb{C}$, let us consider four surjective (onto) functions from nodes $V$ to clocks:
$f_{\mathbb{D}}$, $f_{\mathbb{D}\cup\Delta}$, $f_{\mathbb{C}}$ and $f_{\mathbb{C}\cup\Delta}$ mapping any node $v$ to the clock (among the set in the function's subscript) that maximizes the node improvement. The submodularity property\vsb 
$$I(\mathbb{D} \cup \Delta) - I(\mathbb{D}) \geq I(\mathbb{C} \cup \Delta) - I(\mathbb{C}), \vsb$$
holds due to the following three possibilities for any $v \in V$:\\
\noindent $~\bullet~$ $f_{\mathbb{D}\cup \Delta}(v) \neq \Delta$. This implies that $f_{\mathbb{C}\cup \Delta}(v) \neq \Delta$. As a result $f_{\mathbb{D}}(v) = f_{\mathbb{D}\cup\Delta}(v)$ and $f_{\mathbb{C}}(v) = f_{\mathbb{C}\cup\Delta}(v)$, and $v$ contributes zero to both sides of the inequality. \\
\noindent $~\bullet~$ $f_{\mathbb{D}\cup \Delta}(v) = \Delta$ and $f_{\mathbb{C}\cup \Delta}(v) \neq \Delta$. Then $v$ makes a positive contribution to the LHS and zero contribution to the RHS.\\
\noindent $~\bullet~$ $f_{\mathbb{D}\cup \Delta}(v) =f_{\mathbb{C}\cup \Delta}(v) =\Delta$. Then $v$ contributes positively to both sides. However, $I(v,\Delta) \geq I(v,\mathbb{C}) \geq I(v,\mathbb{D})$ because $f_{\mathbb{C}}$ has access to all of the clocks that $f_{\mathbb{D}}$ does; it will not assign $v$ to a lower-scoring clock. 
\vsb
\end{proof}

\begin{corollary} {\bf [Apx. of \grdp (due to~\protect\cite{submodularityKrause})]}
A greedy selection of $k$ clocks ensures a $(1-1/e)$-apx. for \koc.
\label{cor:apx}\vsa
\end{corollary}

To construct a Greedy solution for kOC we need to be able to find the best next clock $\Delta$ to add to an existing solution of selected clocks $\mathbb{D}$. 
Let $f_{\mathbb{D}}(v)$ be a surjective function mapping every node to the clock $\Delta \in \mathbb{D}$ that maximizes the improvement due to $v$ in all cascades $I(v,\mathbb{X}|\Delta)$. Generalize the improvement of a node given previous clocks $\mathbb{D}$ as: \vsb 
\begin{multline}
     I(v,X,\delta^e_s\leftarrow \delta^{s-1}_b|\mathbb{D})= \\
     \max \{0,I(v,X,\delta^e_s\leftarrow \delta^{s-1}_b) - I(v,X|f_{\mathbb{D}}(v))\},
     \label{eq:intimprmulti}\vsb
\end{multline}
\noindent where $I(v,X|f_{\mathbb{D}}(v))$ is the improvement due to $v$'s activation as part of cascade $X$ in the best clock that $v$ is mapped to in the solutions so far $f_{\mathbb{D}}(v)$. We can similarly generalize the definitions in Eqs.~\ref{eq:totalintimpr}, \ref{eq:optintimpr} to account for already selected clocks obtaining the total interval improvement $I(\delta^e_s\leftarrow \delta^{s-1}_b,\mathbb{X}|\mathbb{D})$ and the optimal interval improvement $I^*(\delta^e_s,\mathbb{X}|\mathbb{D})$ by using $I(v,X,\delta^e_s\leftarrow \delta^{s-1}_b|\mathbb{D})$ in the corresponding definitions. Based on the above, we generalize \dmp($f_{\mathbb{D}}(v)$) to return the clock of best improvement given a selected clocks set $\mathbb{D}$ by passing an additional parameter $f_{\mathbb{D}}(v)$ and replacing all improvement functions in \dmp (Steps 2,6,7,10) with the corresponding improvements given $\mathbb{D}$. The overall \grdp algorithm starts with an empty set of clocks and incrementally adds the next best clock returned by \dmp($f_{\mathbb{D}}(v)$).

\grdp selects the best next clock greedily and is, thus, guaranteed to return a solution of at least $(1-1/e)OPT$ due to Corr.~\ref{cor:apx}. It, however, has a complexity $O(k|\mathbb{X}|^4)$.
As a fast alternative, we also extend \gr (Alg.~\ref{alg:singleGreedy}) to obtain a double-greedy \grsq scheme. It also calls incrementally a version of \gr($f_{\mathbb{D}}(v)$) with the additional parameter that allows for improvement evaluation, conditional on the best node scores from previously selected clocks $\mathbb{D}$ in (Step 7) of Alg.~\ref{alg:singleDP}. 
The complexity of \grsq is $O(k|\mathbb{X}|^2)$ in the worst case.

%% file: tex/5-exp.tex
\vsa
\section{Experimental evaluation}
\vsa
\begin{table}[t]
\centering
\footnotesize
\begin{tabular}{|l|c|c|c|c|c|c|} \hline
           & $|V|$ & $\bar{|E|}$ & $|\mathbb{X}|$ & $\sum|X|$ & min$|X|$ & Duration \\ \hline 
    Synthetic &$1000$&$2200$&$5000$&$185k$&$30$&N/A \\ \hline
    Twitter~\protect\cite{macropol2013act} &$3331$&$1.2m$&$7685$&$117k$&$5$&3/12-4/12\\ \hline
    Flickr-s~\protect\cite{Cha09} &$66k$&$4.2m$&$278$&$200k$&$500$&9/04-3/09\\ \hline
    Flickr-b~\protect\cite{Cha09} &$127k$&$8.4m$&$3176$&$1m$&$200$&9/04-3/09\\ \hline
\end{tabular}
\caption{Datasets used for evaluation.}
\label{tbl:data}
\end{table}

We evaluate the scalability and quality of our methods for optimal clock detection and demonstrate their utility for several applications in both synthetic and real world data. All experiments are performed on a single core conventional PC architecture and implemented in Java (single-threaded).

Evaluation datasets are summarized in Tbl.~\ref{tbl:data}. Our \emph{Twitter} dataset includes propagation of URLs over two months within a snowball subnetwork of the Twitter follower graph including $3.3k$ users~\protect\cite{macropol2013act}. We expand shortened URLs from tweet messages and focus on domains of popular news and commentary websites (e.g. foxnews, breitbart, thehill) as well as photo and video domains, e.g. youtube, instagram, etc. Unique URLs tweeted by at least $5$ nodes form individual cascades resulting in $7.6k$ cascades and $117k$ total activations.

We also employ anonymized data from Flickr including friendship links and favourite photo markings from 
~\protect\cite{Cha09}. We treat each photo as a cascade over the friendship graph initiated by the owner of the photo (first activation) while subsequent users marking the photo as a favorite comprise the rest of the activation events within the photo cascade. We create two datasets for evaluation: all cascades of sizes exceeding $500$ (\emph{Flikr-s}) and all cascades of sizes exceeding $200$ (\emph{Flikr-b}).

We also scale-free structured synthetic networks of at least $1000$ nodes 
and sample IC cascades ($p_n=0.1$, $p_e=0.001$). To simulate the temporal heterogeneity of real-world activations, we randomly \emph{stretch} each original time step to multiple time steps and distribute its activations uniformly in the resulting time steps. This ensures that the relative order of non-simultaneous activations is conserved.

\subsection{Clock discovery: scalability, quality and stability}

\begin{figure}
\centering
\subfigure[][Clocks Twitter (conservative news)]
{
\centering
  \includegraphics[width=0.4\textwidth, trim=70 280 60 280, clip]{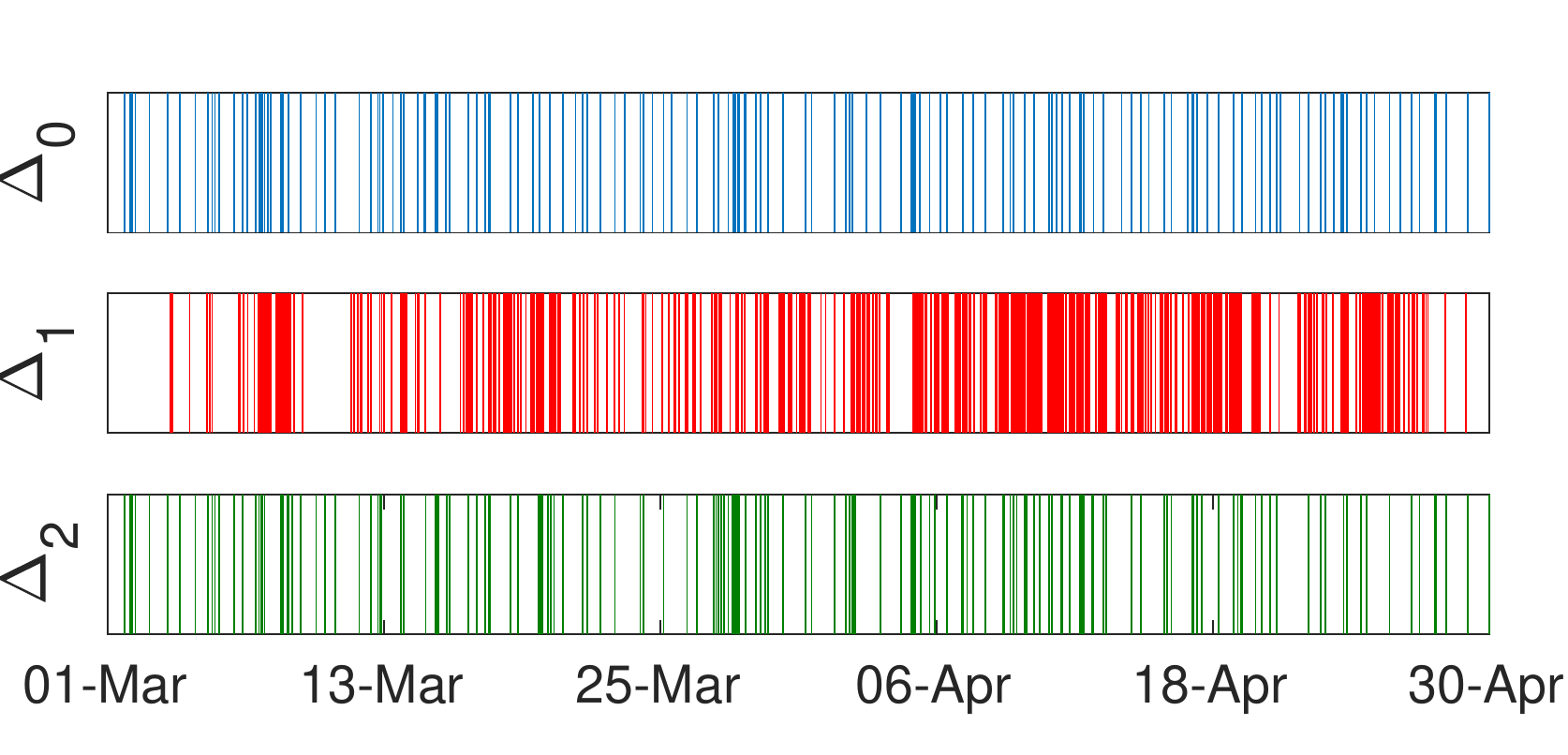}
  \label{fig:ex-clocks-twitter}
}
\vsb
\subfigure[][Clock statistics]
{
\footnotesize

\begin{tabular}[b]{c c c c} 
   &{$|\Delta|$}& {$|V|$}& {$\%I$}\\ \hline 
    $\Delta_0$:& 163 & 1038 & 83 \\ \hline 
$\Delta_1$: & 739 &226& 16 \\ \hline
$\Delta_2$: & 193 & 8 & 1 \\ \hline
\end{tabular}
\label{fig:ex-clocks-twitter-stats}
}

\caption{\footnotesize Example of a \koc solution in Twitter for $500$ conservative media cascades obtained by \grsq. \subref{fig:ex-clocks-twitter}: Interval separators over the period 03/12-04/12 for the first three clocks: $\Delta_0,\Delta_1,\Delta_2$. 
\subref{fig:ex-clocks-twitter-stats}: Statistics of the clocks, their mapping to nodes and \% improvement towards the final solution. 
}
\label{fig:ex-clocks-T}
\end{figure}

\begin{figure*}[ht]
\centering
\subfigure[][Quality]
{
 \centering
  \includegraphics[width=0.22\textwidth]{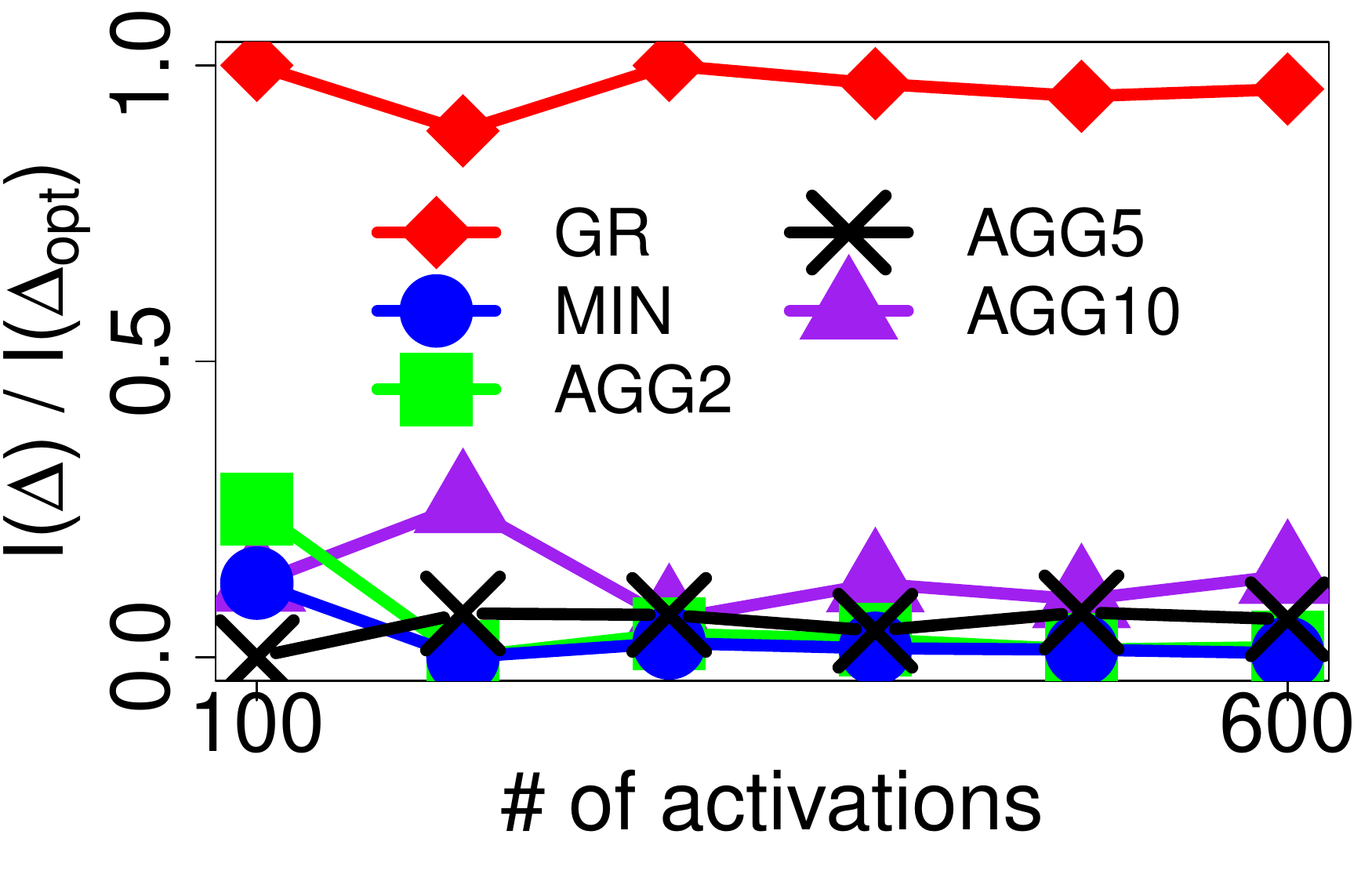}
  \label{fig:q-sin-syn}
}
\subfigure[][Time]
{
\centering
  \includegraphics[width=0.22\textwidth]{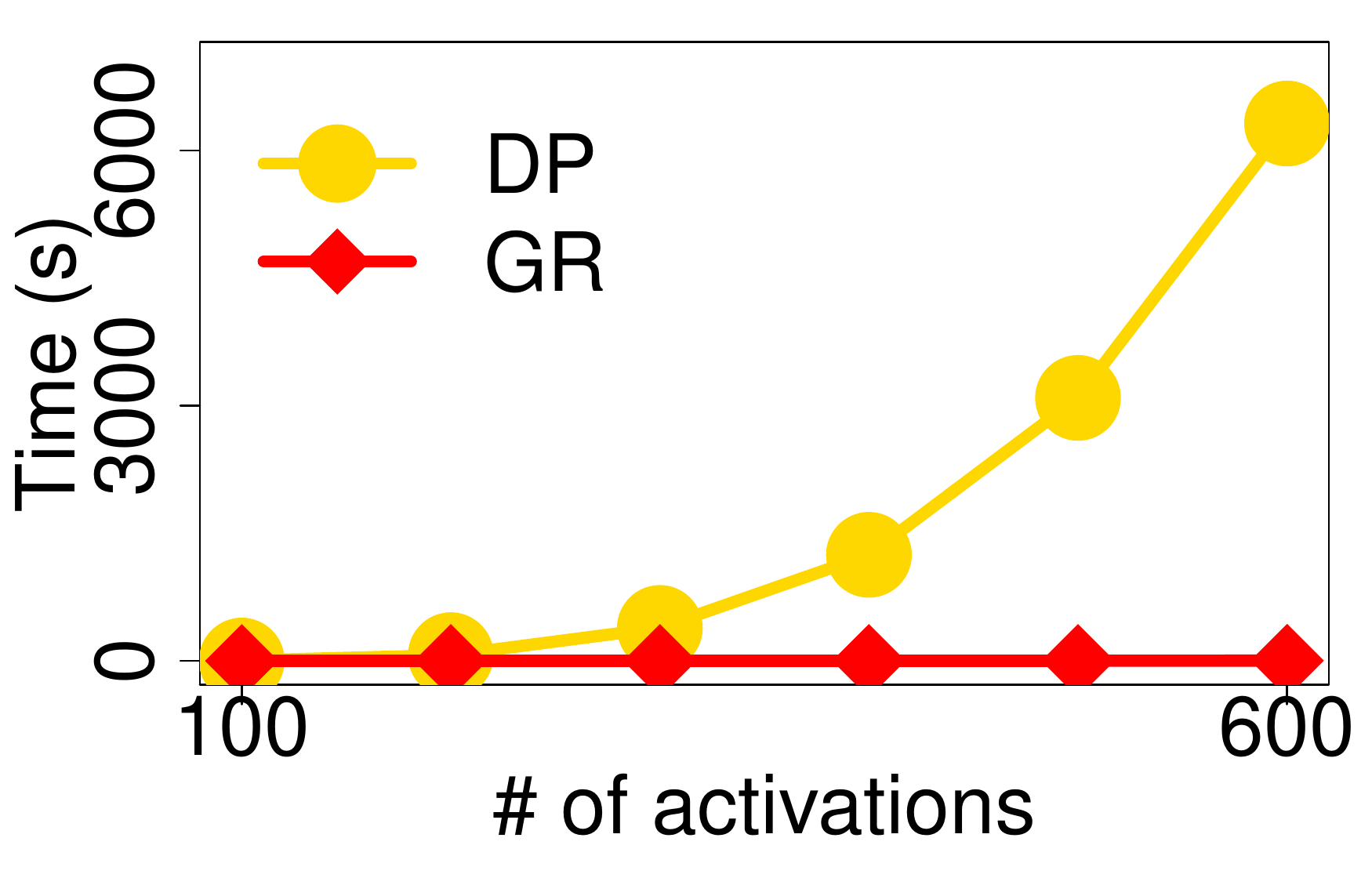}
  \label{fig:t-sin-syn}
}
\subfigure[][Quality kOC]{
\centering\hspace{-0.1in}
  \includegraphics[width = 0.22\textwidth]{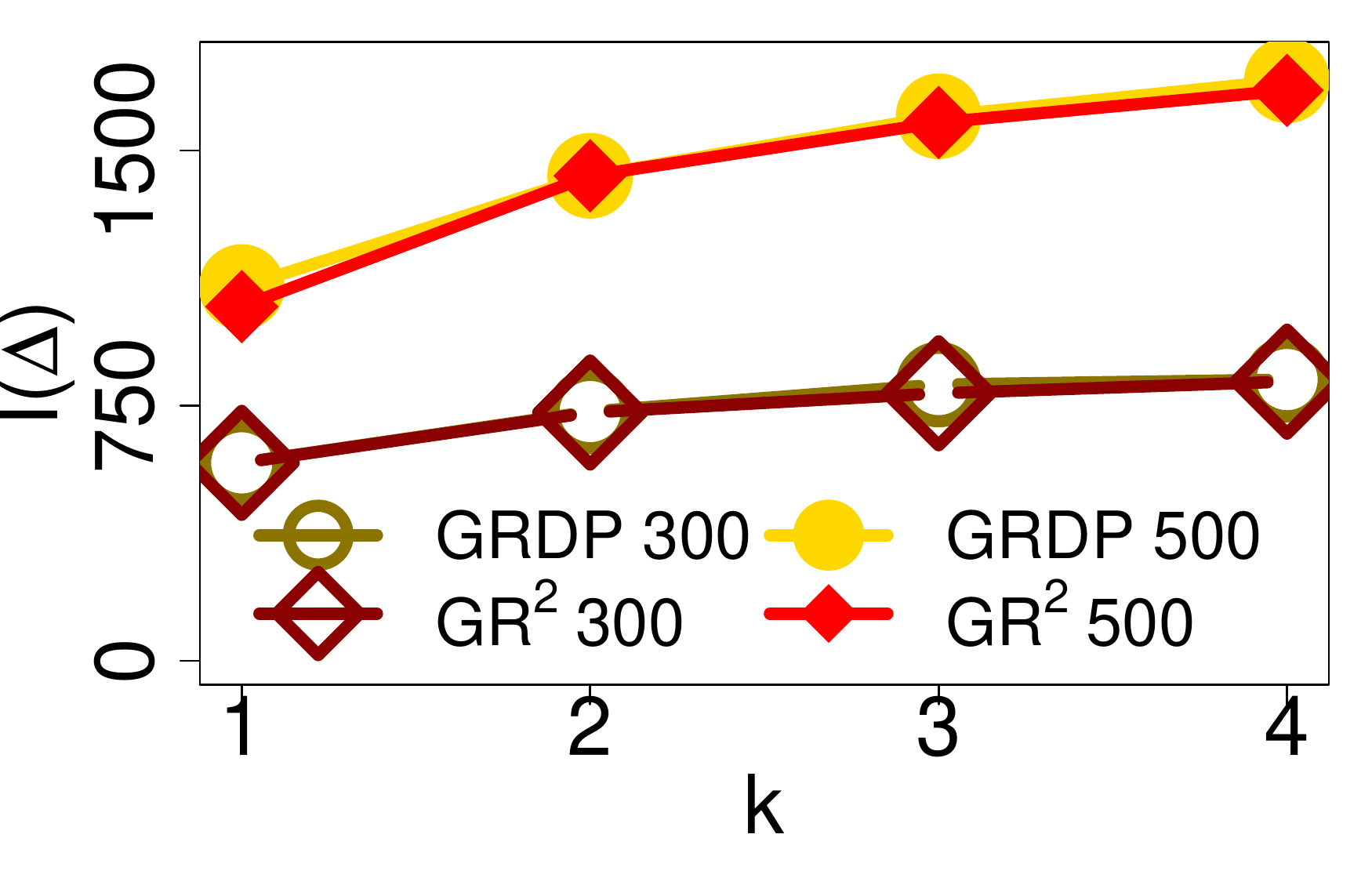}
  \label{fig:q-mul-syn}
}
\subfigure[][Time kOC]
{
\centering
  \includegraphics[width=0.22\textwidth]{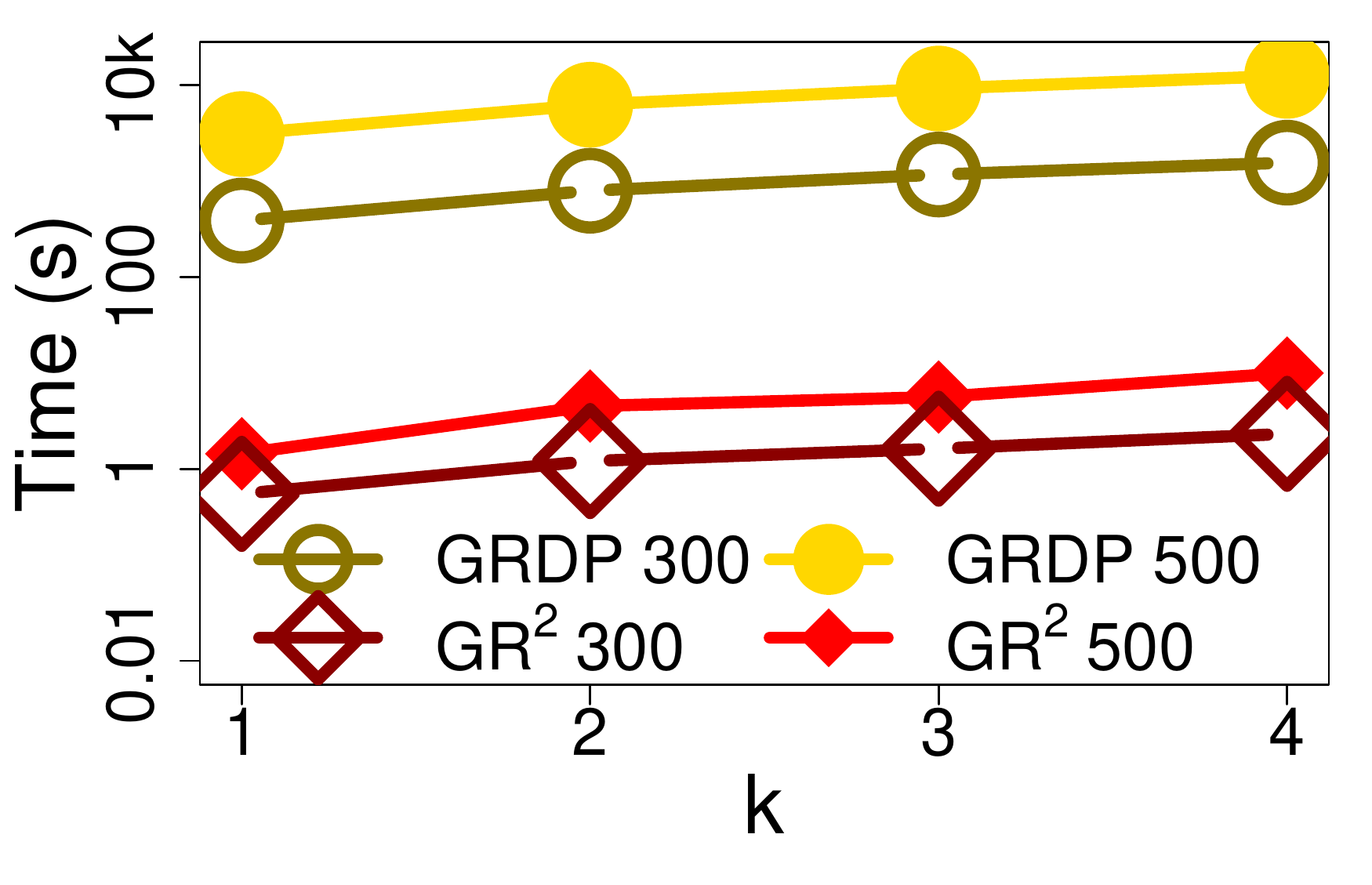}
  \label{fig:t-mul-syn}
}\vsb
\caption{Solution quality and execution time for synthetic data. ~\subref{fig:q-sin-syn} Comparison against the optimal improvement score $I(\Delta_{opt})$ for the clock found by \gr and several homogeneous clocks. \subref{fig:t-sin-syn} Comparison of execution time for \dmp and \gr. \subref{fig:q-mul-syn},\subref{fig:t-mul-syn} Solution quality $I(\Delta)$ and execution time vs number of clocks discovered by \grdp and \grsq on two datasets, one with 300 activations and one with 500 activations. Note that $I(\Delta)$ increases with more activations, as there are more possiblities for improvement over the baseline clock $\Delta_{max}$.}
\label{fig:synthetic-oc}
\end{figure*}

We first evaluate our solutions and discuss obtained clocks in real and synthetic data and than move on to several applications. The clocks obtained by applying \grsq ($k=3$) to the Twitter network are presented in Fig.~\ref{fig:ex-clocks-T}. 
The very first clock $\Delta_0$ exhibits a relatively more regular partitioning of the timeline compared to the following clocks (Fig.~\ref{fig:ex-clocks-twitter}). Intuitively, it captures the best overall time partitioning that best explains all underlying cascades. The following clocks ``improve'' the resolution for activations of certain nodes in the network. About $25\%$ of the nodes get assigned to the second clock $\Delta_1$ and less than a tenth of a percent of the nodes to clock $\Delta_2$ (Tbl.~\ref{fig:ex-clocks-twitter-stats}). Similar diminishing returns are also observed in the percentage of improvement to the overall likelihood contributed by each clock (Col: $\%I$ in Tbl.~\ref{fig:ex-clocks-twitter-stats}). 

\noindent{\bf Scalability and quality.} We investigate the relative quality of the clocks obtained by \gr and various homogeneous aggregations of time compared to the optimal solution $\Delta_{opt}$ found with \dmp for a single clock (\oc) in terms of relative likelihood improvement $I(\Delta)/I(\Delta_{opt})$ in Fig.~\ref{fig:q-sin-syn}. \gr's improvement consistently exceeds $90\%$ of the optimal. We also plot the results for homogeneous time aggregations of fixed windows of $1$ through $10$ original time steps. As expected, no homogeneous clock will work as well because activation rates vary with time, requiring a varying data-driven temporal resolution. When the temporal resolution is too high, the cascades are fragmented; while when it is too low, many activations appear spontaneous. Only by allowing the resolution to change with time (i.e. \oc solutions) do we obtain an effective model of cascade behavior.
\gr's near-optimal solution quality is critical, since \dmp algorithm does not scale to large instances, namely its execution time exceeds $1.5h$ on even modest instance sizes Fig.~\ref{fig:t-sin-syn}. In contrast, \gr scales nearly linearly with the number of activations and discovers high-quality clocks in real-world data with as many as one million activations (see Fig.~\ref{fig:t-sin-real}).

\begin{figure}[t]
\centering
\subfigure[][Flickr-b Quality]
{
\centering
  \includegraphics[width=0.13\textwidth]{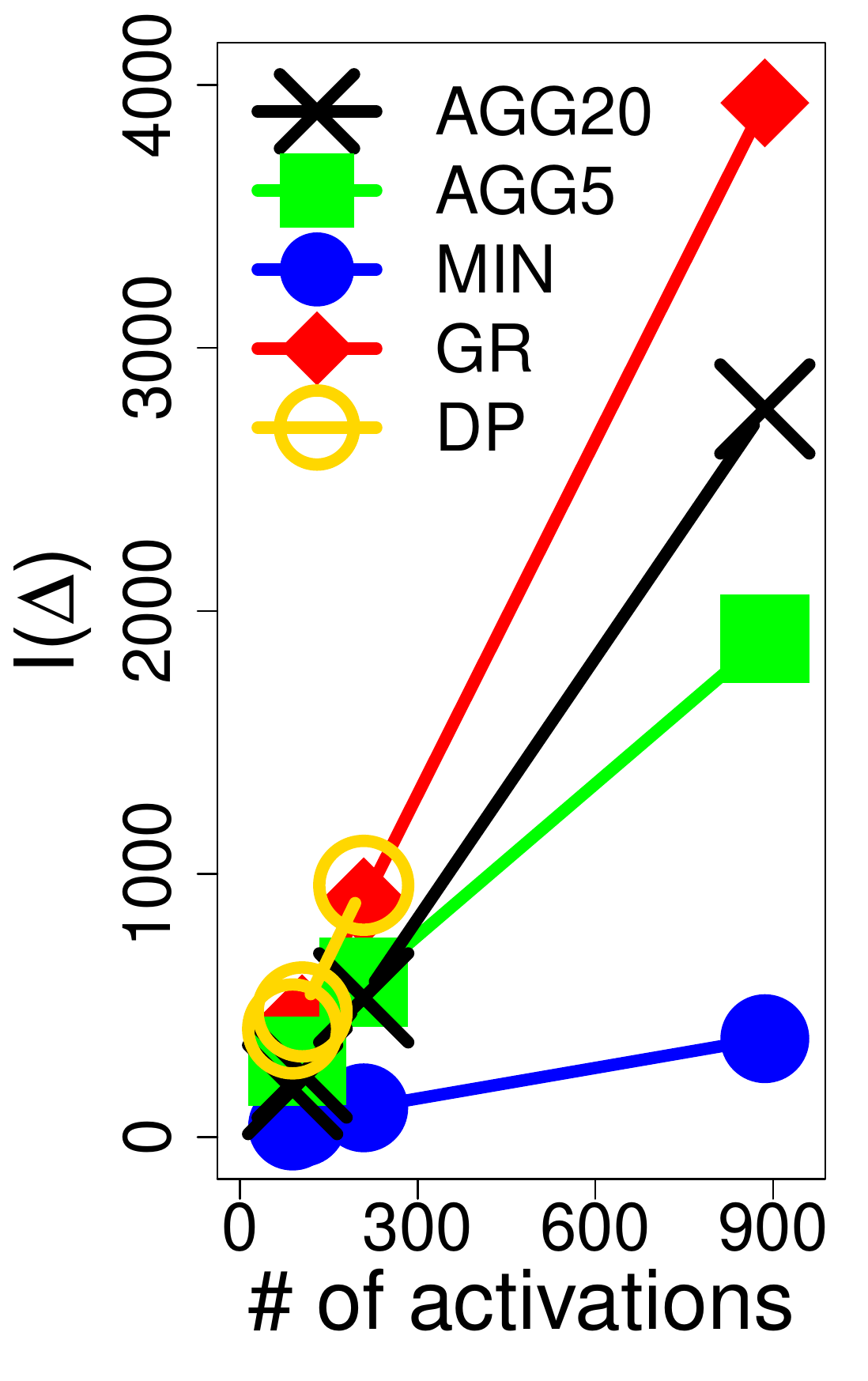}
  \label{fig:q-sin-flickr}
}
\subfigure[][Reg. Clock Quality]
{
\centering
  \includegraphics[width=0.13\textwidth]{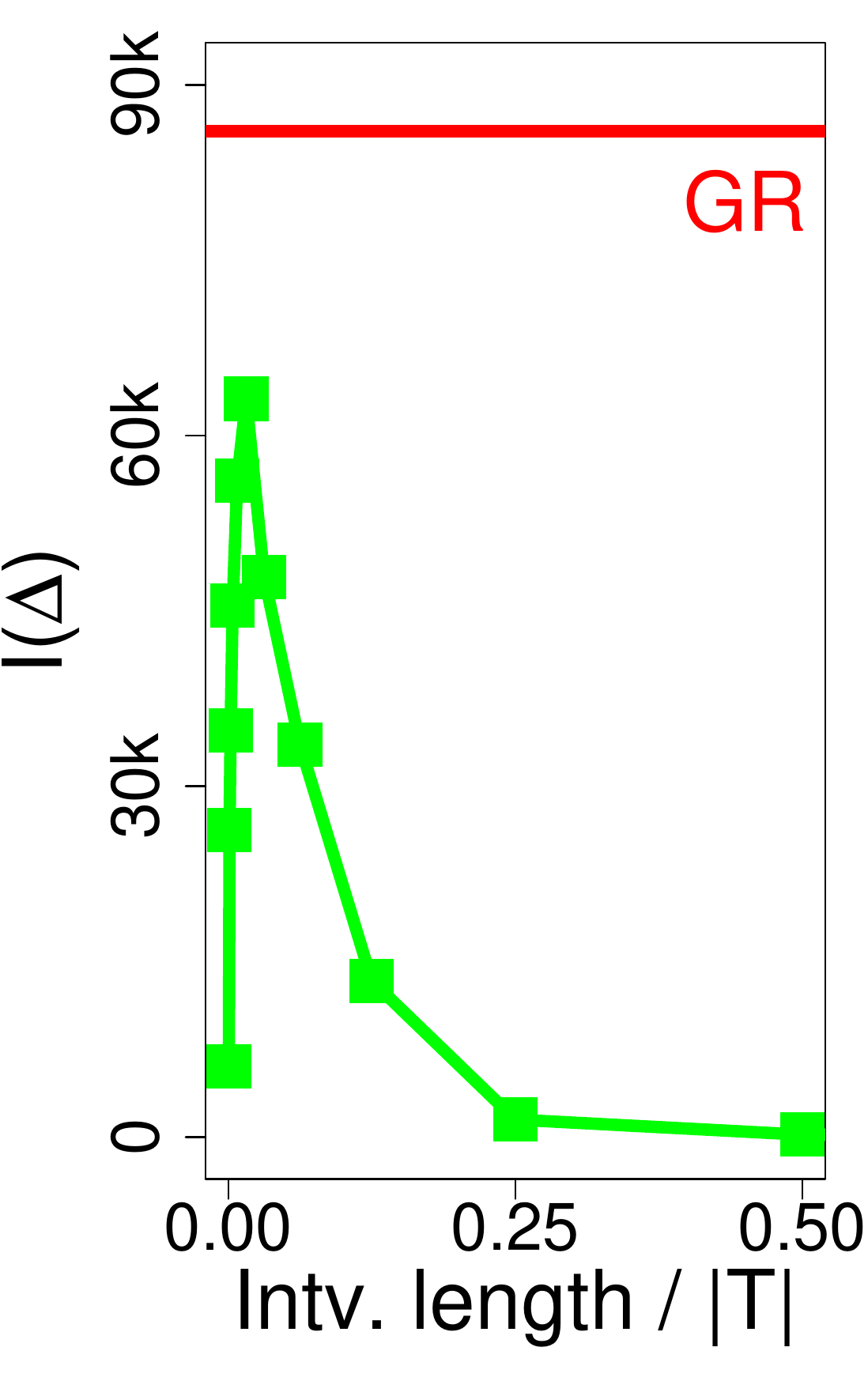}
  \label{fig:reg-cl-flickr}
}
\subfigure[][Time]{
\centering\hspace{-0.1in}
  \includegraphics[width=0.13\textwidth]{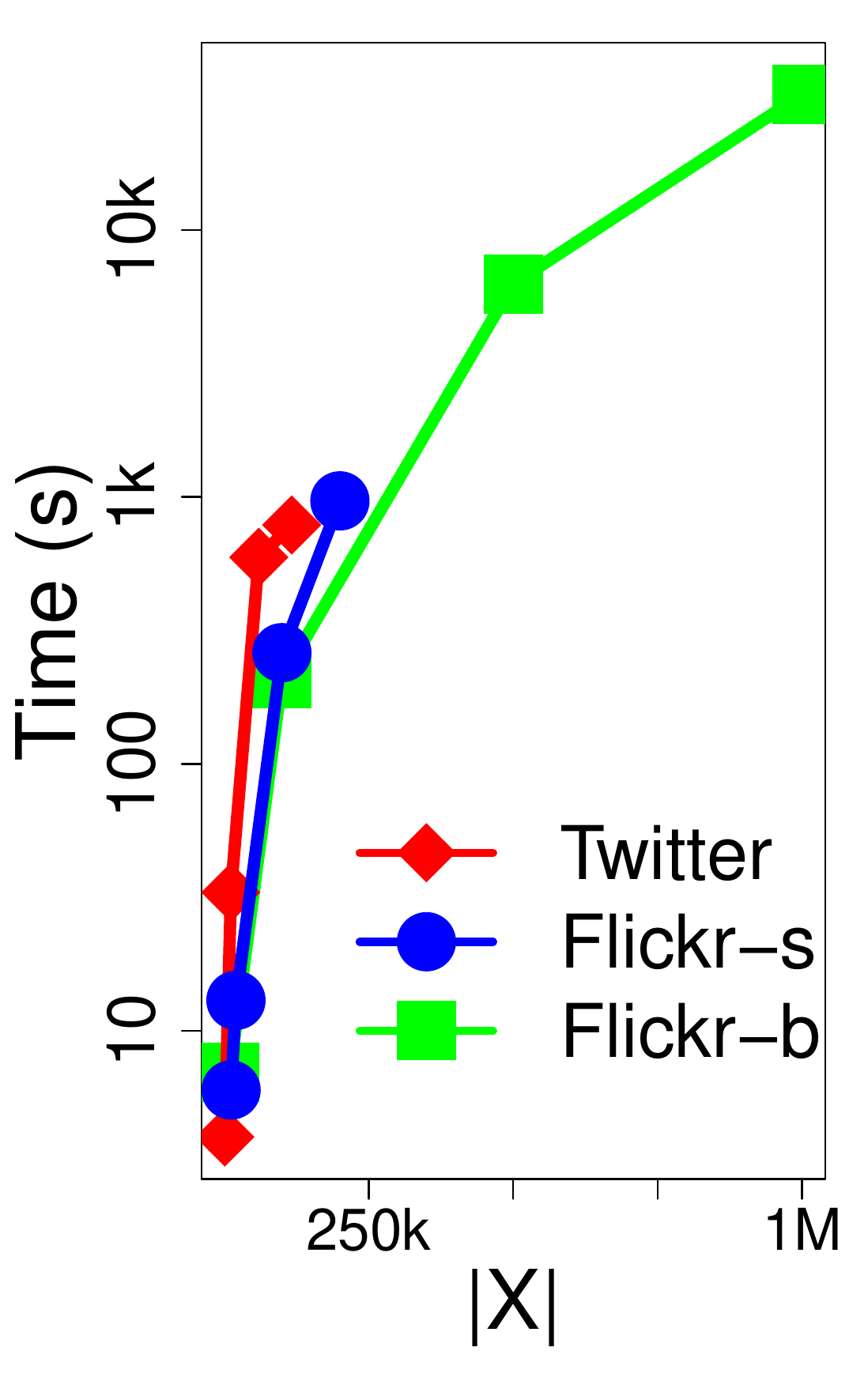}
  \label{fig:t-sin-real}
}
\vsb
\caption{\footnotesize \subref{fig:q-sin-flickr}: Likelihood improvement of \gr compared to \dmp and homogeneous clocks (AGGX, MIN) in Flickr. Similar quality trends persist for \grsq versus \grdp in Twitter and using multiple clocks (omitted due to space limitation).
\subref{fig:reg-cl-flickr}: Comparison of many different homogeneous clocks to \gr using $19k$ activations on Flickr-s. \subref{fig:t-sin-real}: Execution time for detecting a single clock in all three real datasets. Multiple clock detection (omitted) shows an increase in runtime that is linear in $k$.}
\vsc
\label{fig:real}
\end{figure}

In the multi-clock (kOC) setting, we observe similar trends for \grdp and \grsq. While \grdp provides a guarantee of a $(1-1/e)$-approximate solution, its reliance on \dmp limits its scalability to large instances. As observed in Figs.~\ref{fig:q-mul-syn},\ref{fig:t-mul-syn}, the quality of \grsq's solutions closely match those of \grdp with more than an order of magnitude execution speed-up.

\begin{figure*}[ht]
\centering
\subfigure[][Synthetic Success stretch]
{
 \centering
  \includegraphics[width=0.219\textwidth]{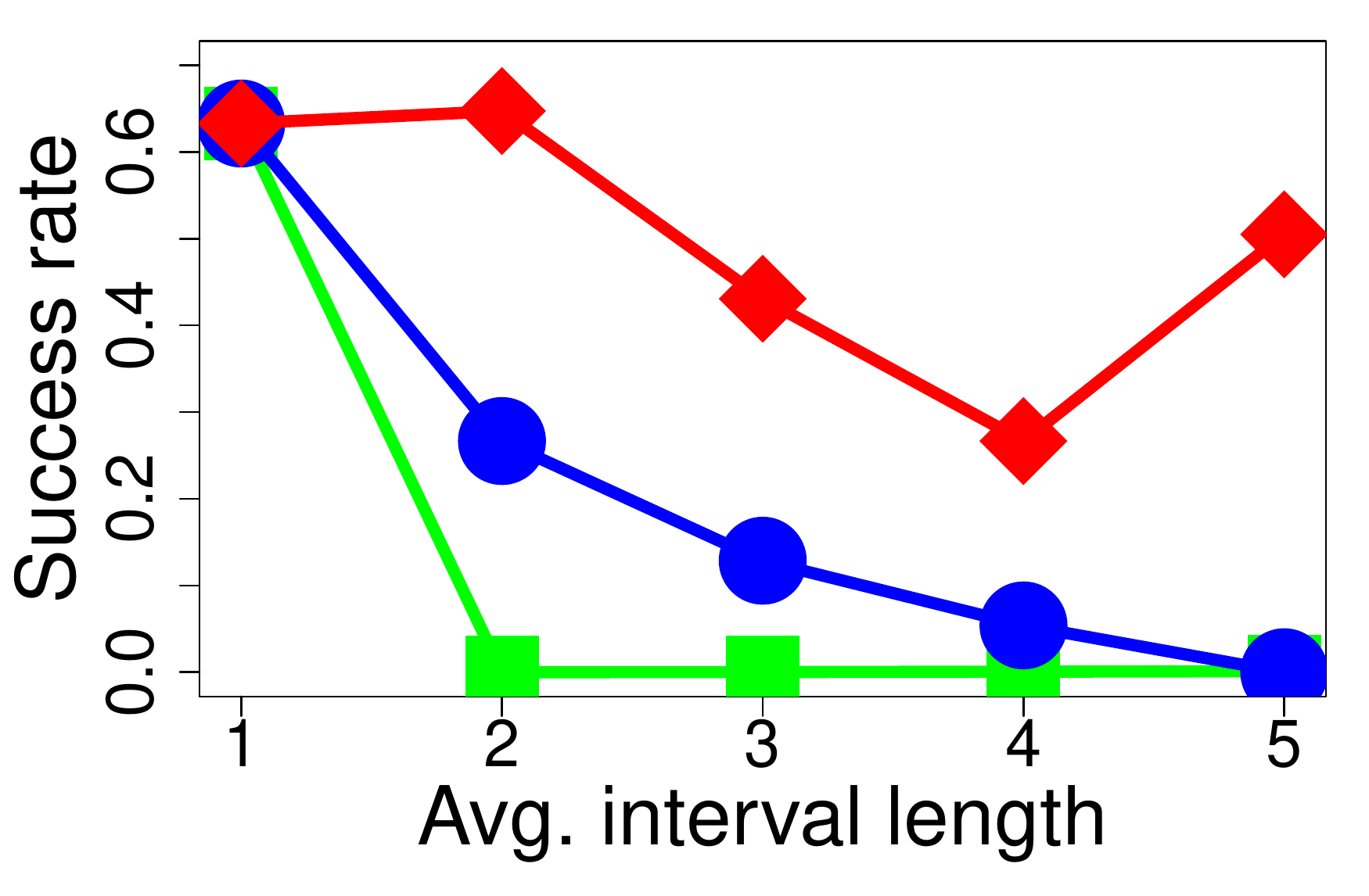}
  \label{fig:succ-stretch}
}
\subfigure[][Synthetic $F_1$ stretch]
{
\centering
  \includegraphics[width=0.219\textwidth]{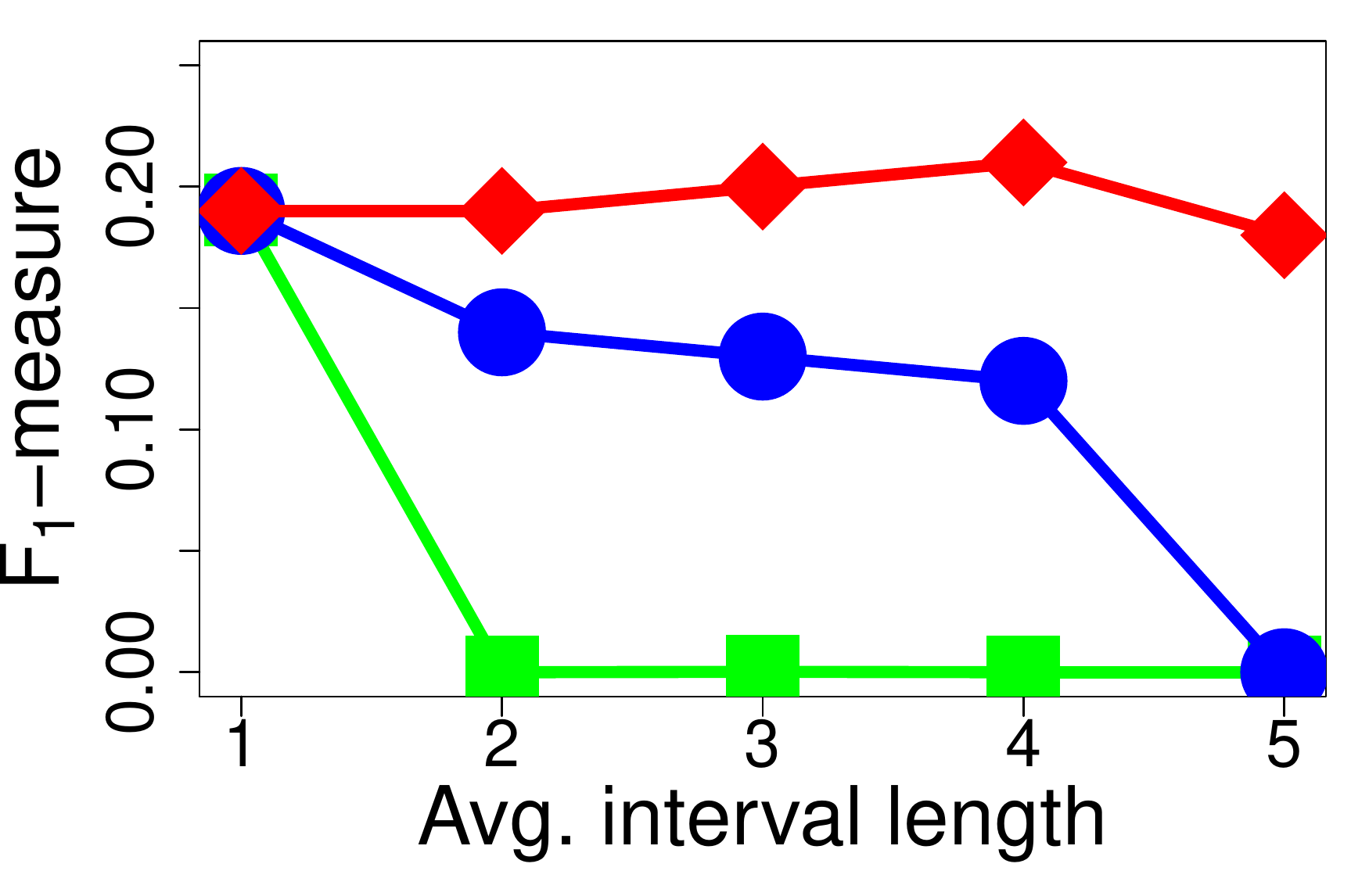}
  \label{fig:f-stretch}
}
\subfigure[][Synthetic success]{
\centering\hspace{-0.1in}
  \includegraphics[width = 0.219\textwidth]{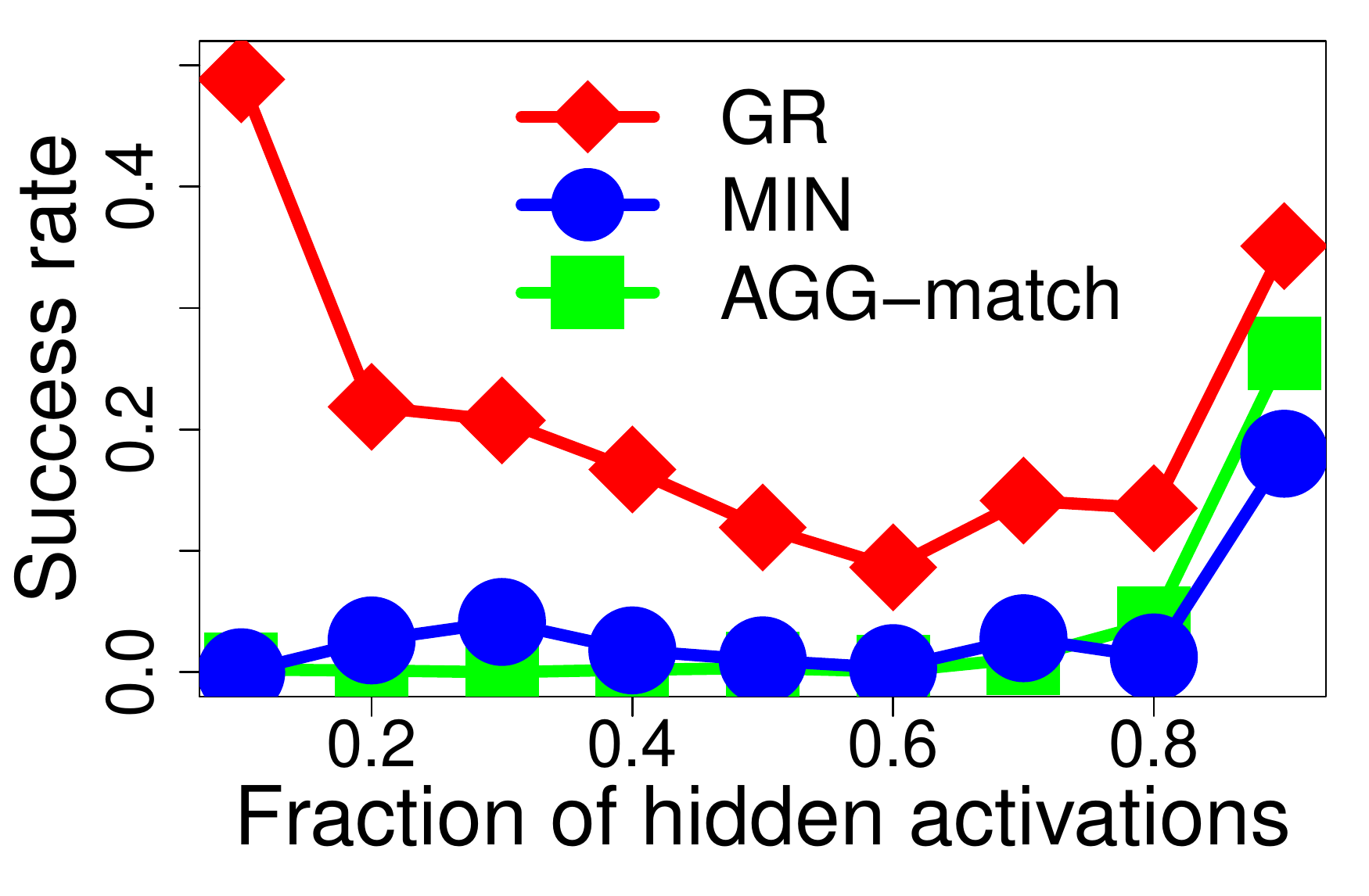}
  \label{fig:succ-drop}
}
\subfigure[][Synthetic $F_1$]
{
\centering
  \includegraphics[width=0.219\textwidth]{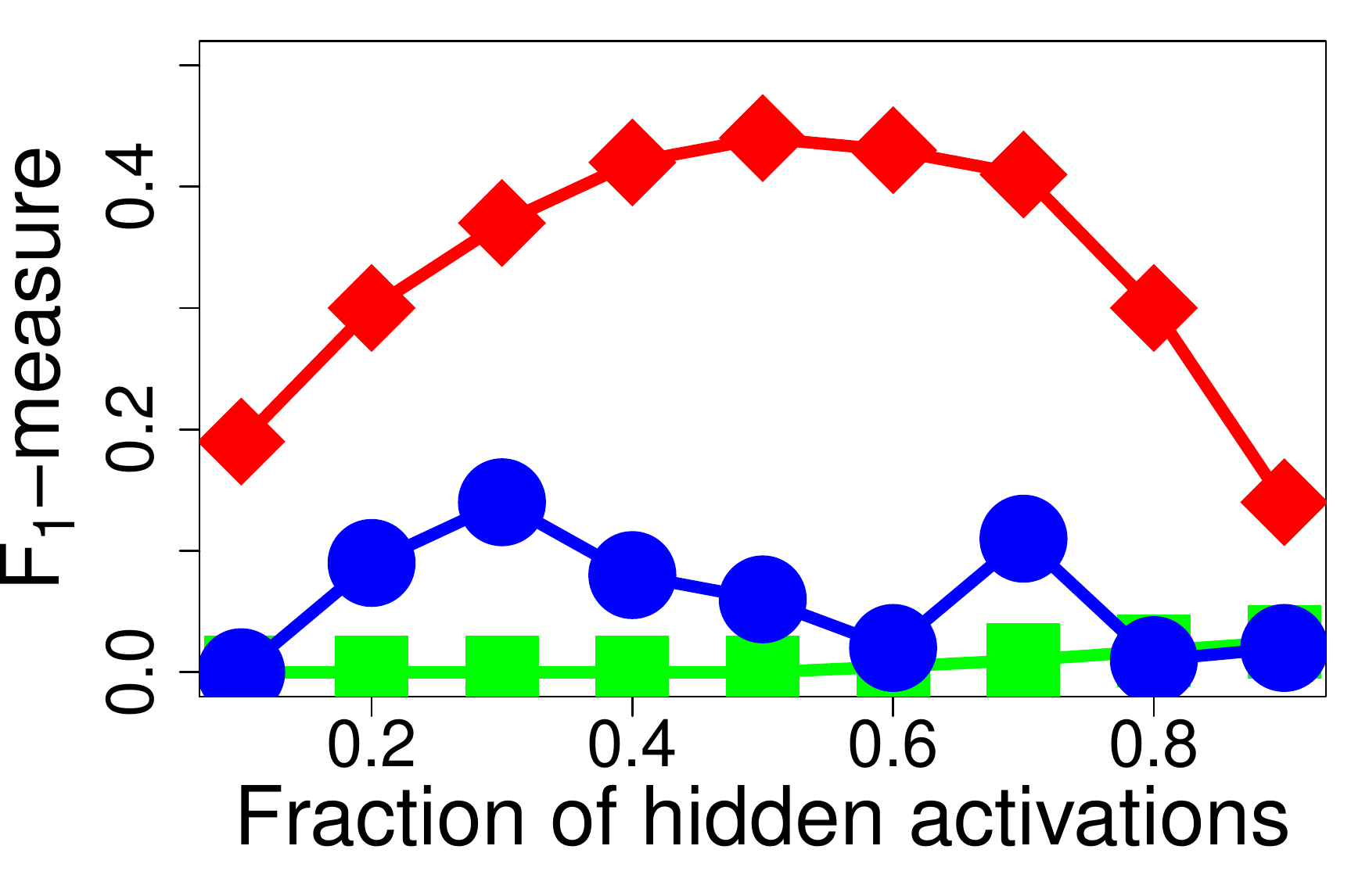}
  \label{fig:f-drop}
}
\subfigure[][Twitter success]{
\centering
  \includegraphics[width = 0.219\textwidth]{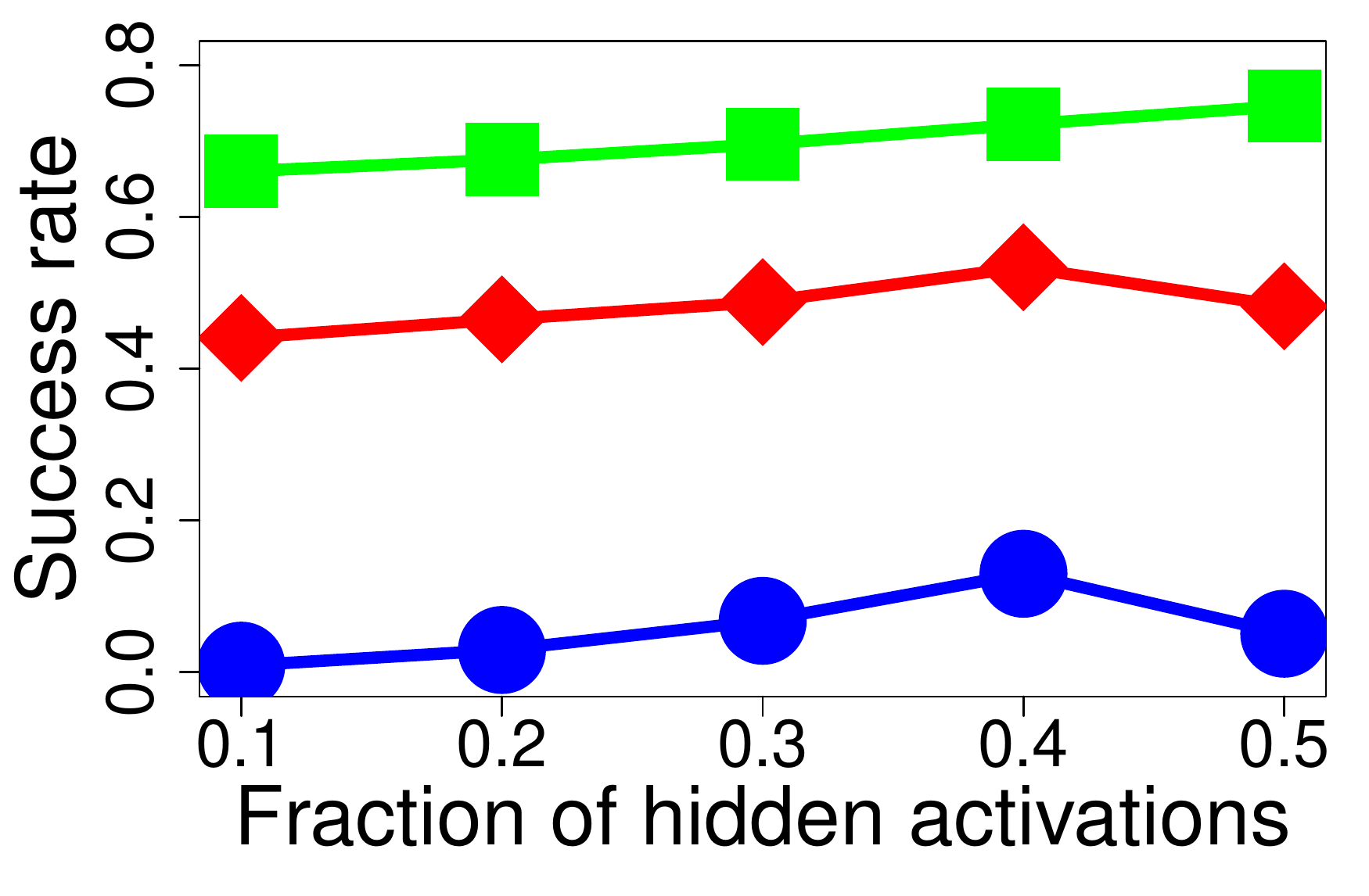}
  \label{fig:succ-drop-twitter}
}
\subfigure[][Twitter recall]
{
\centering
  \includegraphics[width=0.219\textwidth]{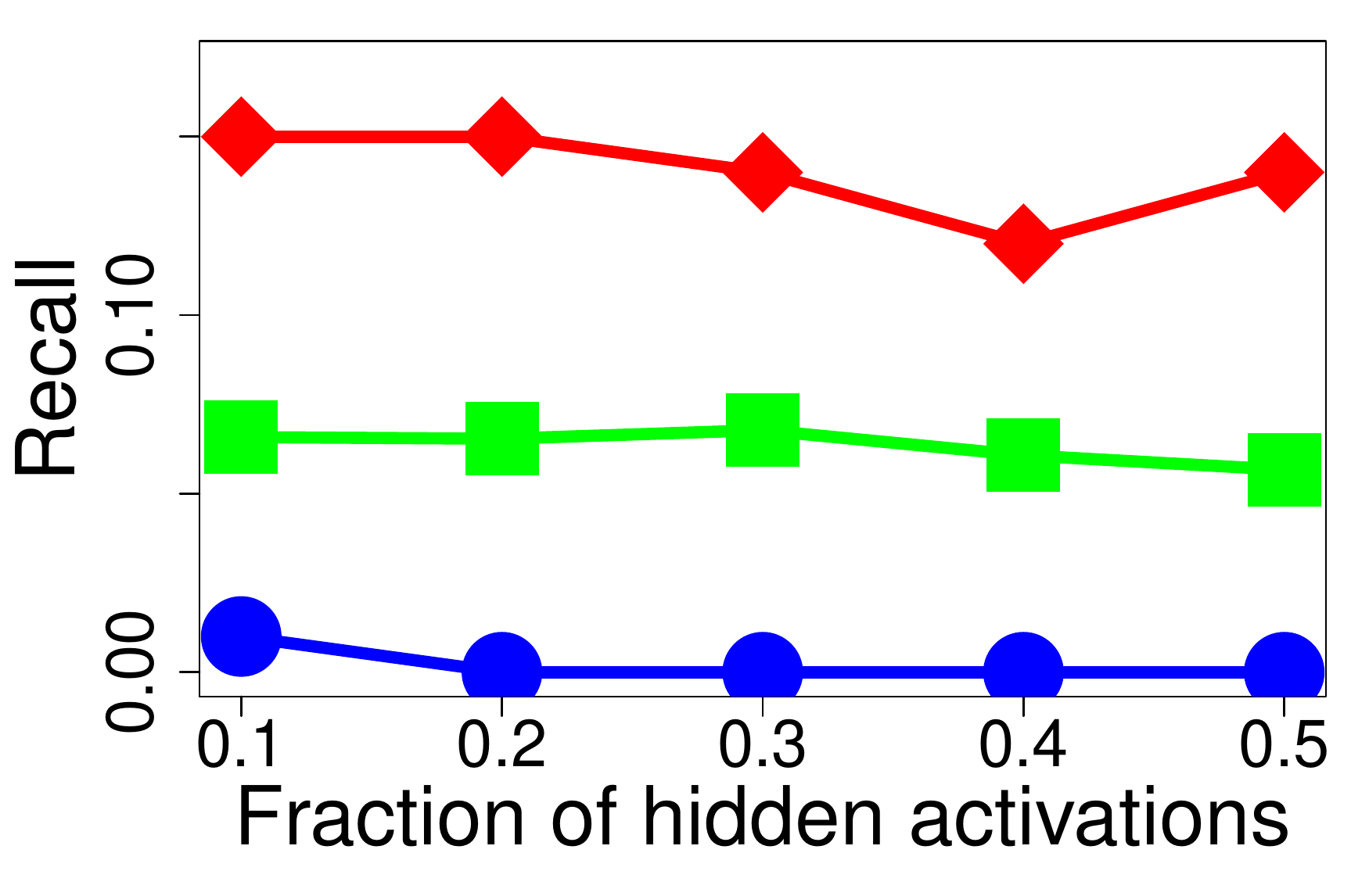}
  \label{fig:recall-drop-twitter}
}
\subfigure[][Flickr success]{
\centering\hspace{-0.1in}
  \includegraphics[width = 0.219\textwidth]{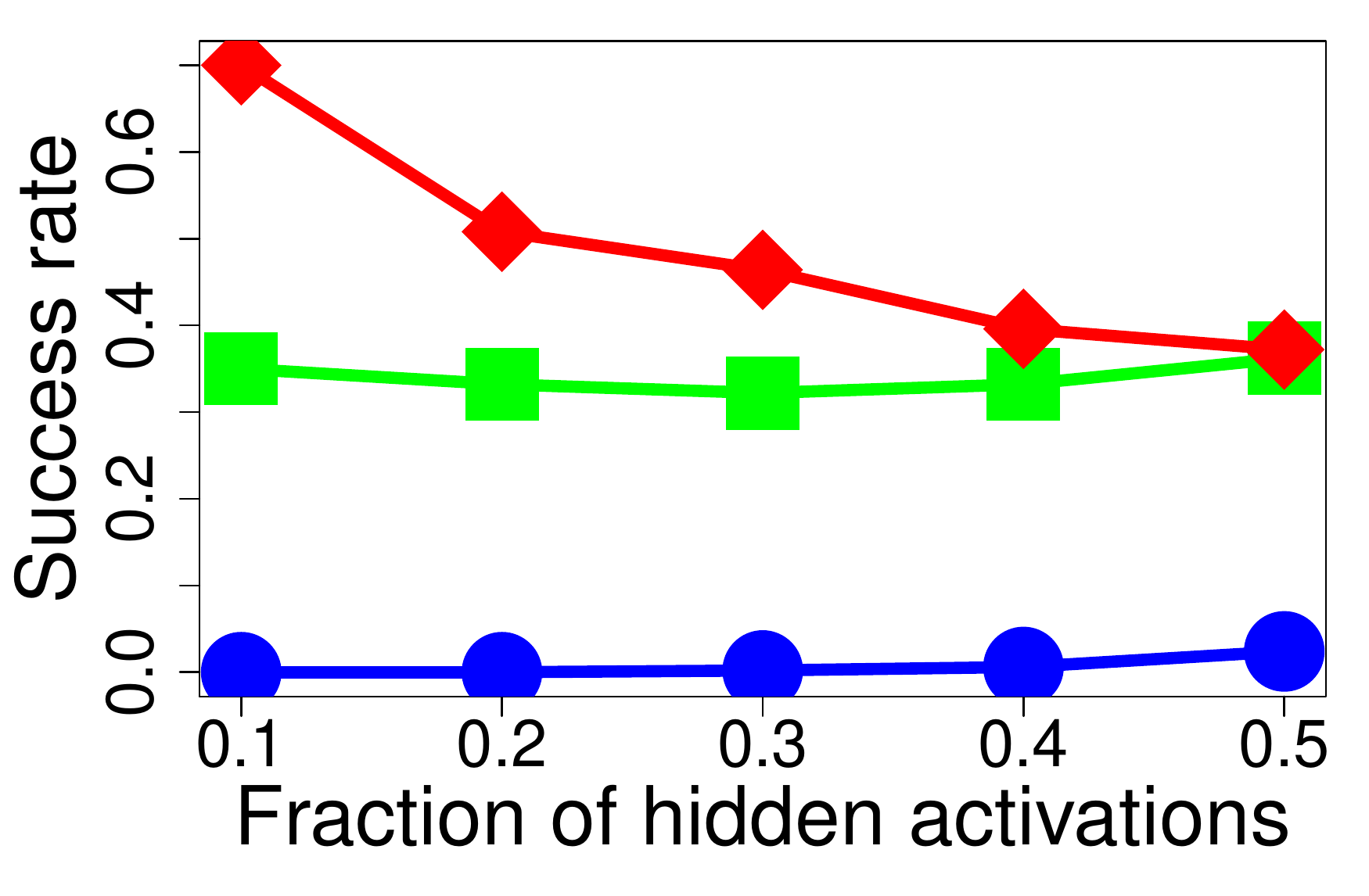}
  \label{fig:succ-drop-flickr}
}
\subfigure[][Flickr recall]
{
\centering
  \includegraphics[width=0.219\textwidth]{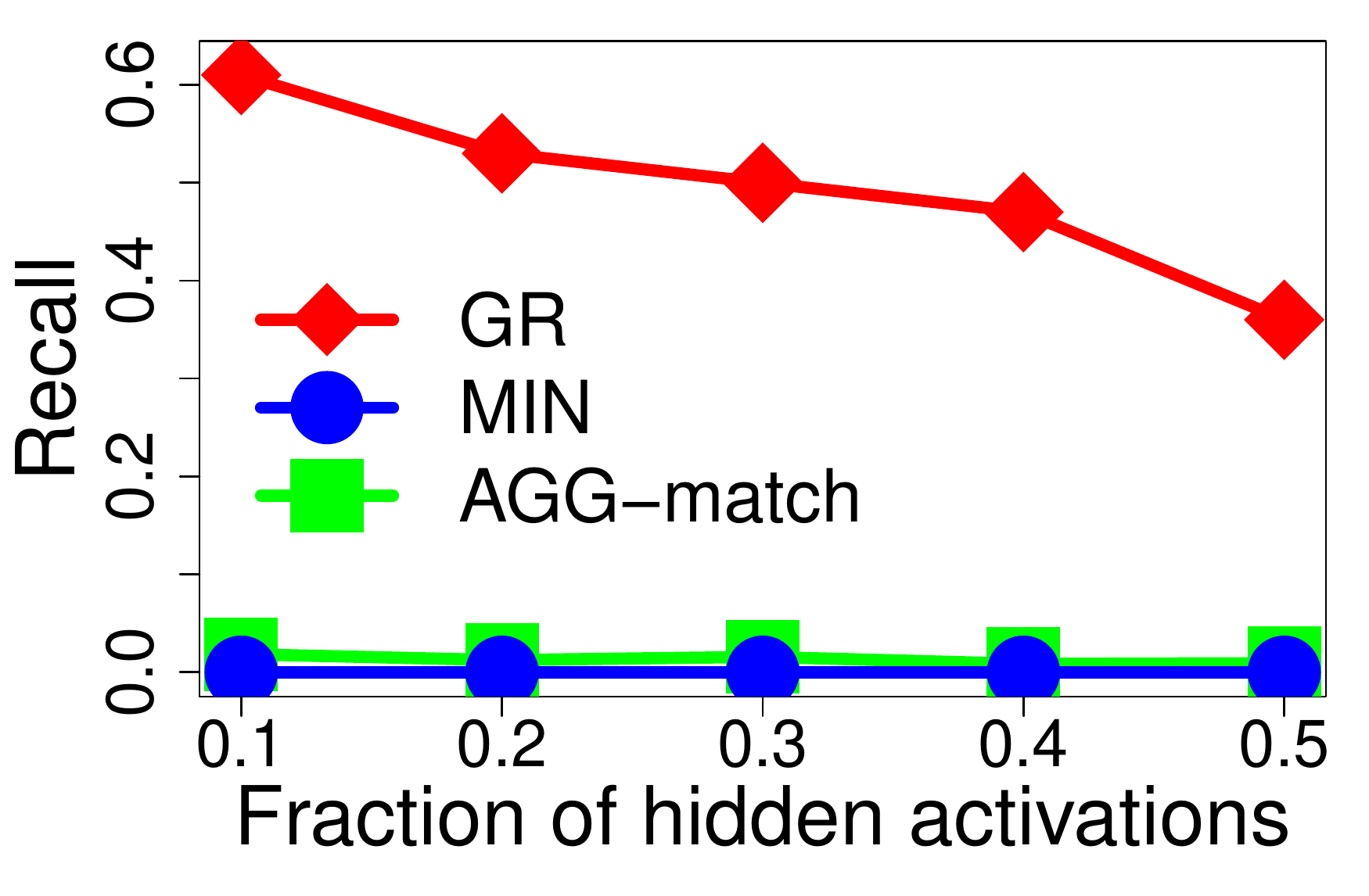}
  \label{fig:recall-drop-flickr}
}\vsb

\caption{Comparison of solution quality for the cascade completion problem in \protect\cite{zong2012inferring}. $MIN$ are results using the original timings ($\Delta_{min}$). $GR$ are results after remapping times according to the clock found by our greedy algorithm. $AGG$-$match$ are results after remapping times onto a homogenous clock with the same number of intervals as the one found by $GR$. (a)-(d) show synthetic data; (e)-(h) show real-world data.}
\label{fig:completion} \vsc
\end{figure*}
We perform similar evaluation on Twitter and Flickr data as well. The quality trends for Flickr-b for a single clock discovery are presented in Fig.~\ref{fig:q-sin-flickr}. Ideally one should compare $I(\Delta)$ to $I(\Delta_{opt})$ (found using \dmp), as we did in Fig.~\ref{fig:q-sin-syn}; however, because these networks are larger, this was impractical for all but the smallest instances. \gr's improvement score closely matches that of \dmp up to $200$ activations and persistently dominates fixed-width aggregations for larger instances. Similar behaviour in terms of quality is observed for Twitter \oc and for both datasets for multiple clock discovery \koc (figures omitted due to space constraints). 

Fig.~\ref{fig:t-sin-real} shows the runtime for \gr on our three real-world data sets. Of particular interest is the fact that Twitter, though a much smaller instance, requires greater execution time for an equivalent number of activations than Flickr. This is because running time of \gr is affected by the network density and our Twitter network is much denser. Sparser data (e.g. using @mentions instead of follows) would likely allow for much faster execution. While the largest size of Flickr-b requires several hours to extract a clock, it is worth mentioning that it processes a million activations in a $8.4m$-edge network with effectively $10s$ of millions of active edges. For kOC, the trends are the similar to those in synthetic data (omitted due to space). Most notably, \grsq remains linear in $k$, since it performs the same call to \dmp for each round of clock discovery.

\noindent{\bf Sensitivity to $p_n$ and $p_e$.} IC' propagation $p_n$ and spontaneous activation $p_e$ probabilities are the only two parameters for our algorithms. Hence, a natural question is: \emph{How stable is the clock w.r.t. these probabilities?} In practice $p_e$ and $p_n$ have minimal impact on the produced clocks for large ranges of their settings. We performed both synthetic and real-world data experiments that systematically vary both parameters, namely $\frac{p_e}{p_n}$ taking values between $0.001$ (i.e. spontaneous activations are very unlikely) and $1$ (i.e. activations are equally driven by external sources and the network). The clocks detected by our algorithms in synthetic data were identical in all cases where $p_e<p_n$; only when the parameters were equal was there a change in the optimal clock. For real-world data, there was more sensitivity with some changes appearing at $p_e \geq 0.2p_n$; however, even in these cases the obtained clocks were largely unchanged. Note that large settings of $p_e$ (comparative to $p_n$) attempts to fit to a diffusion process that is predominantly driven by external forces due to the sheer number of nodes, and hence the network structure's effect is minimal.

\subsection{Applications of optimal clocks} 
Next we apply detected clocks for improved missing cascade data and cascade size and topic prediction.

\noindent{\bf Inferring missing cascade activations.} Due to limitations on API usage of OSNs or privacy restrictions, only a partial observation of the diffusion process are typically available~\protect\cite{zong2012inferring,chen2015full,sadikov2011correcting}. Hence, in order to accurately detect information sources, estimate virality along edges or maximize the influence of campaigns it is important to infer unobserved cascade participants based on the network structure and the partially available cascade information. Zong et al.~\protect\cite{zong2012inferring} study this problem in the context of IC and propose approximate algorithms for the NP-hard problem of optimal cascade reconstruction. 

To illustrate the importance of using an appropriate network clock for this application, we adopt the \emph{WPCT} algorithm~\protect\cite{zong2012inferring} to infer missing activations in both synthetic and real-world data, while comparing the accuracy due to the original and optimal clocks. 
\emph{WPCT} attempts to reconstruct a feasible cascade tree based on partial observed activations. A cascade tree is feasible if the distance from a cascade's source to every observed activation is equal to the time of activation (assuming the source is activated at $t=0$). To allow for cascades with multiple sources, we use a modified version of {\em WPCT\em} that returns a forest so long as all observed activations are reachable from one of the cascade sources. Because this is an NP-hard problem~\protect\cite{zong2012inferring}, the algorithm may fail to find such a forest even when one exists. Working with an appropriate timeline, thus become critical for the algorithm and the completion applications overall.

The \emph{success rate} of {\em WPCT\em} in finding a viable forest (measured as the fraction of cascades that obtained feasible trees) is one measure of performance and the quality of predicted missing activation a second one. To test the algorithm, we begin with complete cascades and then hide activations, removing them from the data provided to {\em WPCT\em}. Each activation in a cascade has a fixed probability of being removed independent of all other activations, except that the first and last activation of each cascade are always kept since the algorithm has no chance of discovering them. Considering these hidden activations as true positives and any other non-observed nodes as false positives, we measure precision and recall for obtained solutions.

To generate synthetic data for this experiment, random cascades with at least 30 activations were created on a fixed, scale-free network. To emulate the delays in a real-world network, these cascades were then stretched: a random clock with a known average interval length was generated, and then each cascade had its times remapped onto that clock. For example, if the third interval in a clock was $\intv{5}{8}$, then an activation at $t=3$ would be remapped randomly onto a time between $5$ and $8$, inclusive. The results on synthetic data are shown in Figs.\ref{fig:succ-stretch}-\ref{fig:f-stretch}. Each figure shows three traces: the result after applying the clock found by \gr, the result after applying a homogeneous clock with the same number of intervals as the \gr clock (AGG-match), and the baseline (MIN) using the original timeline ($\Delta_{min}$). Using the \gr clock greatly increases the success rate when cascades have been stretched, which would be the case when the temporal resolution is too high. It also outperforms both alternatives for any rate of hiding activations from the algorithm. The spike in success rate at extremely high hidden-activation rates (Fig.~\ref{fig:succ-drop}) stems from the removal of constraints. To complete a cascade tree, one of simply needs to recover any path from source to the few remaining activations of appropriate length. The prediction of missing nodes (F1 measure), however, decreases for large drop rate due to the large number of possible candidates Fig.~\ref{fig:f-drop}. Essentially, while the recall for such cases is still high, precision decreasses drastically.

In real-world data, we only vary the probability of hiding activations. Both networks are quite dense, resulting in huge number of candidates for completion and thus high false positive rates in all cases; hence, we only plot recall instead of $F_1$-measure. The Twitter network is so dense that the homogeneous clock outperforms \gr on the success rate measure (Fig.~\ref{fig:succ-drop-twitter}; there are so many possible paths that at a matching level of temporal resolution (same number of intervals as in a \gr clock, but homogeneous in length), it becomes very likely that a feasible tree can be found. Most of these solutions are, however, incorrect, as indicated by the much lower recall rate of AGG-match than that of \gr in Fig.~\ref{fig:recall-drop-twitter}. In Flickr's less-dense network, \gr outperforms both alternatives on both metrics (Figs.~\ref{fig:succ-drop-flickr}-\ref{fig:recall-drop-flickr}).

\begin{table}
\centering
\footnotesize
\begin{tabular}{|l|c|c|c|c|} \hline
    {\bf Data}& {\bf m}& {\bf F1 w/o Clock}& {\bf F1 w Clock}& {\bf \#casc.}\\ \hline 
    \multirow{3}{*}{{\bf Twitter}} &$10$&$0.67$&$0.67$&$1280$\\ 
        &$20$&$0.66$&{\bf0.71}&$664$\\ 
        &$30$&$0.7$&{\bf0.76}&$309$\\ \hline \hline
    \multirow{3}{*}{{\bf Flickr-b}} &$300\dots500$&$0.66$&$0.66$&$1177\dots271$\\ 
        &$600$&$0.62$&{\bf0.69}&$173$\\ 
        &$700$&$0.61$&{\bf0.67}&$94$\\ \hline
\end{tabular}
\caption{Cascade size prediction aided by detected clocks.} \label{tbl:sizep}\vsb
\end{table}

\noindent{\bf Predicting cascade size} based on initial observations is
\label{apl:size}
another relevant application we consider. Recently Cheng et al~\protect\cite{cheng2014can} considered multiple features of the $m$-prefix of observed cascades (the first $m$ activations) to predict the eventual size of cascades.
The authors observed that temporal features 
have highest impact on prediction accuracy~\protect\cite{cheng2014can}. Naturally, precisely these features will be most affected by inappropriate temporal aggregations of the timeline. For example, in cascades initiated late at night, observing a large time-lag to the morning activations might make the cascade prefix appear slow-growing and bias its prediction to limited future growth. Hence, we investigate the potential for improvement of size prediction based on learned network clocks. 

We employ only the temporal features proposed in~\protect\cite{cheng2014can} and set-up the same prediction problem: given the first $m$ activations, how likely is it that a cascade exceeds size $\alpha m, \alpha\geq 1$? We set $\alpha$ such that the considered cascades are balanced between the \emph{small} ($|X|<\alpha m$) and \emph{large} ($|X|\geq\alpha m$) class. In~\protect\cite{cheng2014can}, $\alpha$ was set to $2$ as this produced balanced prediction problems, however, the appropriate setting for our datasets is $\alpha\approx 1.5$, as our cascade size distributions have a different slant. We compare the F1-measure in cross-validation of the temporal features according to (i) the original times of activation (col. 3 in Tbl.~\ref{tbl:sizep}) and the same according to the clock obtained by \gr (col. 4 in Tbl.~\ref{tbl:sizep}), considering prefixes of increasing size $m$ in Twitter and Flickr.

In Twitter the optimal clock enables $7\%$ lift of the F1 measure for $m=20$ and $8\%$ for $m=30$, and performs on par with the original clock for shorter prefixes. A significant improvement due to the optimal clock in Flickr occurs for prefixes of size $600$ and higher. The cascades in Flickr tend to be very large and a big number of the original activations are disconnected in the network and hence do not allow the clock to boost the informativeness of the features. This might be due to many popular photos being featured, thus inducing non-network driven favorite markings. The optimal clock provides a significant prediction improvement ($m\geq 600$) and if combined with possible other features (tags,text) it might aide to even better performance for actual prediction. 

\noindent{\bf Topical consistency of clocks.} Do similar topics and information types spread similarly in the network? In this application, we are interested to quantify the level of consistency of the optimal clocks for information of the same type. We use clock-based cascade topic prediction quality as a proxy for clock consistency. We prepare two classification tasks listed in Tbl.~\ref{tbl:topicp}. In the first case we select URLs from liberal (Huffington post and the hill) and conservative (Breitbart and Foxnews) outlets such that the total number of activations match in each class and we predict (using leave-one-out validation) the accuracy of placing a cascade in its own category simply based on the likelihoods of the cascades according to the two clocks detected based on the remainder of the categories. We repeat the same for news URLs (Text) versus Visual content (photo/video).    

The accuracy in both settings is significantly above random ($50\%$ will be random due to the balanced number of activations per class and, thus, similar size clocks). This observation alludes to different content spreading according to different conserved clocks---a direction we plan to pursue in the future after obtaining a larger annotated corpus of simultaneous cascades. Similar evaluation is not possible in Flickr as the data is anonymized. It is also important to note that the clock likelihood is not the best predictor of cascade topic (e.g. content and users profiles would likely yield much higher quality), however, we use this task only as a proxy for clock consistency within topics/content. 
\begin{table}
\centering
\footnotesize
\begin{tabular}{|l|c|c|c|c|} \hline
   {\bf Comparison}& {\bf \#acts. per class}& {\bf min$|X|$}& {\bf \#casc.}& {\bf ACC}\\ \hline 
   { Liberal v.s. Conservative} &$5301$&$30$&$55$&$56.2\%$\\ \hline
   { Text v.s. Visual} &$15132$&$20$&$342$& $57.5\%$\\ \hline
\end{tabular}
\caption{Cascade topic prediction as a proxy for clock consistency.}
\label{tbl:topicp}
\end{table}

%% file: tex/11-related.tex
\vsb
\section{Related Work}
\vsa
To the best of our knowledge, we are the first to propose algorithms for model-driven {\em heterogeneous\em} time aggregations of dynamic network data on information cascades. Previously, varying time windows have been proposed by Yang et al~\protect\cite{yang2014finding} to mine progression stages in event sequences; however, this work assumes independent sequences and no network structure involved, while we are interested in network diffusion processes. The effect of {\em homogeneous\em} temporal aggregations of the timeline in dynamic networks is a better-studied ~\protect\cite{krings2012effects,caceres2011temporal,sulo2010meaningful,budka2012predicting}. Different fixed window sizes and their effect of the resulting smoothness of time series of network statistics was consisdered in ~\protect\cite{sulo2010meaningful,caceres2011temporal}.  
Significant differences in the structural features of networks formed by aggregating mobile phone call records at different resolutions was observed in~\protect\cite{krings2012effects}, and the effect of aggregation window size on predicting changes in the network was studied in~\protect\cite{budka2012predicting}.

All above problems clearly demonstrate the need to select an optimal aggregation for evolving networks. However, the focus in all of them is on smoothness of global or local network statistics and the proposed approaches select a homogeneous time scale, i.e. each aggregation window is the same length of time. We demonstrate in experiments that such homogeneous aggregations are inadequate for explaining observed information cascades, while our in-homogeneous network clock solutions are superior in terms of likelihood of observed cascades and for several applicatios.

Other approaches for cascade completion were also considered~\protect\cite{sadikov2011correcting}, as well as the problem of missing data when trying to {\em limit\em} cascades~\protect\cite{budak2011limiting}. Providing information about the optimal clock to such algorithms could potentially improve their overall performance similar to the ones we considered. Other applications to which clock data could be applied include influence maximization~\protect\cite{kempe2003maximizing} and estimation of the diffusion probabilities of IC~\protect\cite{saito2008prediction}.


%% file: tex/6-conclude.tex
\vsb
\section{Conclusion}
\vsa
We proposed network clocks: heterogeneous partitions of the timeline that best explain observed information propagation cascades in real-world networks with respect to the independent cascade model (IC). 
We showed an optimal dynamic programming solution \dmp for the single optimal clock (\oc) problem and a scalable greedy alternative \gr that matches (within $10\%$) \dmp's quality, but is able to detect clocks in an  $8.4$-million-edges Flickr network and more than a million activations of nodes within cascades. We demonstrated that the multi-clock problem (\koc) is NP-hard and proposed both a $(1-1/e)$-approximation \grdp and a fast and accurate heuristic \grsq. The detected clocks enable improved cascade size classification (up to $8\%$ F1 lift) and improved missing cascade data inference ($0.15$ better Recall) in Twitter and Flickr, compared to using the original time or fixed-window aggregations of time. 